\documentclass[11pt]{elsarticle}
\usepackage{geometry}
\usepackage{setspace}
\usepackage{array}
\renewcommand{\arraystretch}{1.2}
\usepackage{lineno,hyperref}
\modulolinenumbers[5]
\usepackage{subcaption}
\usepackage{float}
\usepackage{caption}
\usepackage{diagbox} 

\usepackage{booktabs}
\newcommand{\ra}[1]{\renewcommand{\arraystretch}{#1}}
\usepackage{changepage}
\usepackage{comment}
\usepackage{multirow}
\usepackage{longtable}
\usepackage{soul}

\makeatletter
\def\ps@pprintTitle{%
	\let\@oddhead\@empty
	\let\@evenhead\@empty
	\let\@evenfoot\@oddfoot
}
\makeatother

\journal{ }

\biboptions{authoryear}

\usepackage{amsmath}
\usepackage{amsfonts}
\usepackage{amsthm}
\usepackage{amssymb}
\usepackage{mathrsfs}
\usepackage{graphicx}
\usepackage{adjustbox}

\newtheorem{proposition}{Proposition}[section]

\usepackage{tikz}
\usetikzlibrary{bayesnet}
\usepackage{xcolor}
\definecolor{dark_green}{rgb}{0.0, 0.5, 0.0}

\DeclareMathOperator{\KL}{KL}

\DeclareMathOperator{\softmax}{softmax}

\DeclareMathOperator{\diag}{diag}

\begin{document}

\setlength{\aboverulesep}{0pt}
\setlength{\belowrulesep}{0pt}
\begin{frontmatter}

\title{Embedded Topics in the Stochastic Block Model}

\author[UP]{Rémi Boutin}
\ead{remi.boutin.stat@gmail.com}

\author[Nice]{Charles Bouveyron}
\author[UCA,UP]{Pierre Latouche}

\address[UP]{Université Paris Cité, CNRS, Laboratoire MAP5, UMR 8145, Paris, France}
\address[Nice]{Université Côte d'azur, INRIA, MAASAI Team, Sophia-Antipolis, France}
\address[UCA]{Université Clermont Auvergne, CNRS, Laboratoire LMBP, UMR 6620, Aubière, France}

\begin{abstract}
	Communication networks such as emails or social networks are now ubiquitous and their analysis has become a strategic field. In many applications, the goal is to automatically extract relevant information by looking at the nodes and their connections. Unfortunately, most of the existing methods focus on analysing the presence or absence of edges and textual data is often discarded. However, all communication networks actually come with textual data on the edges. In order to take into account this specificity, we consider in this paper networks for which two nodes are linked if and only if they share textual data. We introduce a deep latent variable model allowing embedded topics to be handled called ETSBM to simultaneously perform clustering on the nodes while modelling the topics used between the different clusters. ETSBM extends both the stochastic block model (SBM) and the embedded topic model (ETM) which are core models for studying networks and corpora, respectively. The inference is done using a variational-Bayes expectation-maximisation algorithm combined with a stochastic gradient descent. The methodology is evaluated on synthetic data and on a real world dataset.
\end{abstract}


\end{frontmatter}

\newpage


\section{Introduction}
Many real life interactions induce the exchange of texts, as in co-authorship networks, social networks or emails for instance. Since the storage capacity keeps increasing, networks with textual data on the edges become even more frequent. In order to make such networks, called communication networks, intelligible to humans, it is of great interest to gather information about the texts exchanged between the nodes and to summarise the connectivity structure. While those two questions have been studied independently, the work we propose aims at bridging the gap between the two by modelling the joint distribution of texts and edges. To the best of our knowledge, the interest on making the two disciplines of topic modelling, when texts are present on the edges, and model-based graph clustering meets is recent and the methods that have been proposed only rely on the frequency of word within the documents without incorporating semantic meaning. In this paper, we propose to take advantage of pre-trained word embeddings in the topic-model as presented in \cite{dieng2020topic} in order to incorporate semantic meaning of the words and to obtain topic-meaningful clusters.

\section{Related work}
Both the topic modelling methods and the graph clustering techniques have first emerged as deterministic optimisation problems to progressively incorporate uncertainty which led to many developments in the statistical literature. The next part provides a brief summary of the advancements in those domains.

\subsection{Probabilistic models for topic modelling}
 The statistical analysis of topics has emerged in the late 90s with \cite{Papadimitriou98latentsemantic}, developing statistical results for the latent semantic indexing (LSI), first proposed by \cite{deerwester1990indexing}. LSI relies on a spectral analysis of the ``term frequency - inverse document frequency" and successfully captures synonymy between words. To overcome the lack of probabilistic foundations of LSI, \cite{Hofmann1999plsi} introduced the probabilistic latent semantic index (pLSI) which models each word distribution as a mixture model such that each mixture component corresponds to a ``topic". The topic membership of each word is modelled by a multinomial random variable in pLSI. Even though the topic membership of the words depends on the document, a major drawback of this model is the absence of model at the document level. This was overcome by \cite{blei2003latent} with the latent Dirichlet allocation (LDA) which for each document uses a Dirichlet random variable to model the proportion of each topic. However, the Dirichlet distribution makes the topics almost uncorrelated and does not directly model correlation. \cite{blei2006correlated} then proposed to use a normal-logistic prior instead of a Dirichlet prior on the topic proportion to directly model the correlations. All these models require to derive the equations for any new generative model. In \cite{Srivastava2017AutoencodingVI}, they bridged the gap between topic modelling and autoencoders, taking full advantage of gradient descent for those models. Nevertheless, all the former approaches do not incorporate semantic meaning to the words. Indeed, since the model is only based on the document term-frequency matrix, they loose the information provided by the order of the words. In the embedded topic model (ETM), \cite{dieng2020topic} used the strength of word embeddings, such as the continuous bag of words (CBOW) or skipgram \citep{mikolov2013efficient} as a part of the decoder of a variational autoencoder (VAE). The topics are also embedded into the same vector space which allows to easily measure similarities between words and topics. The optimisation is done using gradient descent, as proposed in \cite{rezende2014stochastic} or \cite{kingma2014autoencoding}. For a review of the former methods relying exclusively on the document term frequency matrix, the reader may refer to \cite{vayansky2020review}.

\subsection{Probabilistic models for graph analysis} Statistical network analysis first started with random graph theory, initiated by \cite{erdos1960evolution}. They studied probabilistic properties of graphs with binary connections, and a unique probability for any connection to exist. However, real life datasets do not show such regularity. Therefore, more complex and realistic  graph structures have been considered. Here, a structure designates a partition of the nodes such that nodes in a cluster present a homogeneous connectivity pattern. For example, a community is a group of nodes highly connected one to another but with few connections to the rest of the graph. If the graph is only composed of communities, reordering the adjacency matrix by group would output a block matrix. Another direction emerged with \cite{fienberg1981categorical} who first introduced a probabilistic model that assumes that the probability for two nodes to be connected only depends on the group to which they belong to and applied it to Sampson's monastery dataset \citep{sampson1969crisis}. Introducing a latent representation of the nodes then became popular thanks to the latent position cluster model \citep{handcock2007model} or the stochastic block model (SBM) \citep{Wang1987StochasticBF, nowicki2001estimation}. Many extensions have been developed to incorporate valued edges, as in  \cite{mariadassou2010uncovering}, as well as categorical edges in \cite{jernite2014random} or to add prior information in \cite{zanghi2010vertexfeatures}. Some developments also focused on looking for overlapping clusters \citep{airoldi2008mixed, latouche2011overlapping} as well as dynamic networks \citep{matias2017statistical, zreik2017dynamic, corneli2016block}. The inference of SBM-based model is often done either using Markov chain monte carlo (MCMC), variational expectation maximisation (VEM) as in \cite{daudin2008mixture} or variational Bayes expectation maximisation (VBEM) as in \cite{latouche2012variational}. The classification can either be deduced from the latent variable distribution or be incorporated in the optimisation strategy with a hard clustering, for instance using the classification variational expectation maximisation (CVEM) algorithm \citep{bouveyron2018stochastic}. The choice of the number of cluster $K$ can either be done through a model selection criterion \citep{daudin2008mixture, latouche2012variational}, through a greedy search \citep{come2015} or through a non parametric schemes \citep{kemp2006learning}. Fore more insights about SBM developments, see \cite{lee2019review}. For reviews on statistical network modelling, we also relate to \cite{goldenberg2010survey} and \cite{matias2014modeling}.

\subsection{Probabilistic models for the joint analysis of texts and networks}
The rise of data combining networks with texts, such as emails, social networks or co-authors articles led to developing methods using both the network and the textual information. In that regard,  \cite{zhou2006probabilistic} proposed the community-user topic model (CUT). This model relies on the author-topic model (AT) \citep{Rosenzvi2004author} and adds a latent variable to the bayesian hierarchical model for modelling the communities. Two versions are proposed in the paper, CUT1 hypothesises that a community is entirely defined as a group of users while CUT2 makes the assumption that a community is defined as a set of topics. This model is inferred using a Gibbs sampler to approximate the joint distribution of the communities, topics and users. Eventually, the community-author-recipient-topic (CART) model introduced in \cite{Pathak08socialtopic} makes use of communities both at the document generation level and at the author and recipient generation level which corresponds to the network generation.  However, the high number of parameters combined with the inference based on a Gibbs sampler does not allow to scale this model to large datasets. The topic-link LDA presented in \cite{liu2009topic} also offers a joint-analysis of texts and links in a unified framework by conditioning the generation of a link on both the topics within the documents and the community of authors. The inference relies on a variational EM approach which allows to scale to large datasets but this method only deals with undirected networks. Finally, the topic-user-community models (TUCM) was introduced in \cite{sachan2012using} and was able to discover topic-meaningful communities. The main feature of this model was its capacity to incorporate different types of interactions, well-suited for social networks applications (tweets, retweets, messages, comments, ...). The inference relies on Gibbs sampling which can be limiting when dealing with large datasets. The stochastic topic block model (STBM) presented in \cite{bouveyron2018stochastic} was the first model to handle the simultaneous clustering of nodes and edges while keeping the inference tractable to large dataset thanks to a variational classification EM based inference. This model was extended in \cite{berge2019latent} for the simultaneous clustering of rows (observations) and columns (variables). It was also adapted for dynamic networks in \cite{corneli2019dynamic}. Unfortunately, those models only rely on word counts and cannot use the position of words within a sentence or any form of context information.

\subsection{Our contribution}

In this paper, we propose a new methodology called the embedded topics in the stochastic block model (ETSBM), to look for node partitions incorporating the connectivity patterns as well as the topics exchanged between the nodes. We will reserve the term \textit{community}  to groups of nodes that are densely connected together but poorly connected to the rest of the graph. In the block model literature, the term \textit{cluster} denotes a group a nodes that share a similar connectivity pattern which goes beyond the concept of community. For instance, contrary to communities, a star pattern is defined by two clusters with low intra-connection and large inter-connection probabilities \citep{latouche2012variational}. Such pattern is particularly common in social networks. This type of cluster cannot be retrieved by community detection methods. In this paper, we will also assume that the nodes of a same cluster share a similar use of topics proportions. To find clusters complying with this definition, \textbf{(i)} we propose a generative model assuming that each node belongs to a cluster and that the probability of connection between two nodes, as well as the topic proportions of a document, only depend on the clusters of the corresponding nodes. Figure \ref{fig:scenario_C} illustrates the necessity to combine graph clustering and topic modelling in order to distinguish all four clusters and to obtain more meaningful topics for each cluster. \textbf{(ii)} To model the topics exchanged between the nodes, the documents are encoded with a deep neural network to benefit from their flexibility. \textbf{(iii)} The decoder is made of word and topic embeddings, as in \cite{dieng2020topic}. \textbf{(iv)} In this work, the documents are aggregated at the cluster level, into $Q^2$ meta-documents with $Q$ the number of clusters. The meta-documents are obtained by weighting each document with the cluster membership probabilities of the corresponding nodes. In particular, our inference strategy is able to directly optimise the construction of the meta-documents through the inference procedure.

\begin{figure}[htp]
	\centering
	\includegraphics[width=.3\textwidth]{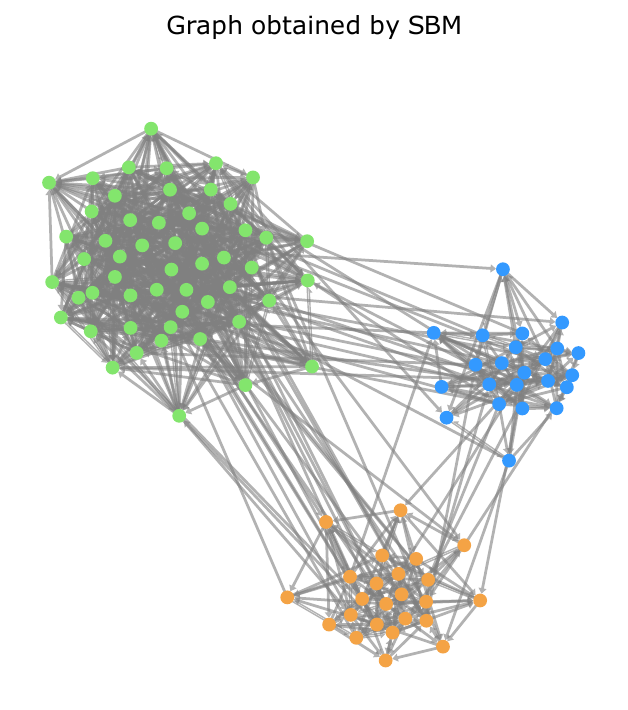}\hfill
	\includegraphics[width=.3\textwidth]{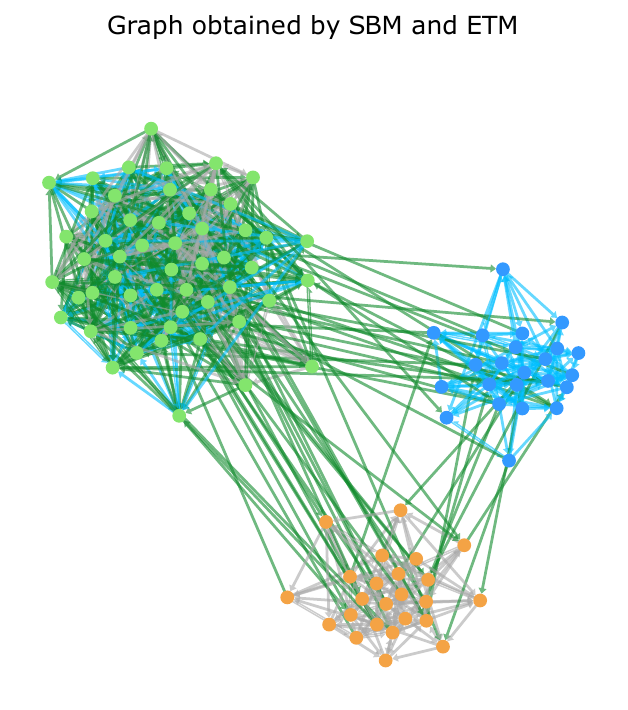}\hfill
	\includegraphics[width=.3\textwidth]{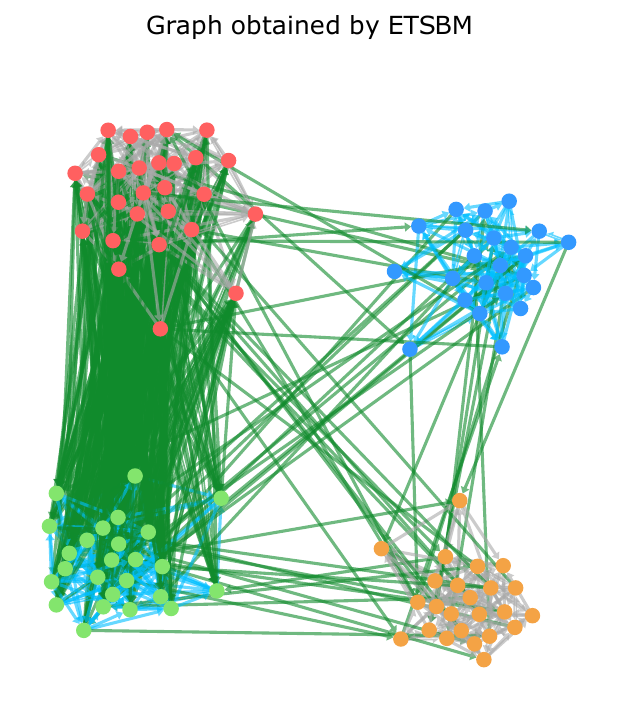}
	
	\caption{Comparison of results on a simulated network with the use of SBM on the left, SBM and ETM in the middle and ETSBM on the right. The colours of the nodes indicate the cluster of the vertices.  The colours of the edges indicate the most-used topic in the corresponding documents. Note that SBM alone does not provide edge information. Thus, the left network only has a single edge colour. On the left hand side, SBM clustering results uncover 3 clusters. Again, in the middle, SBM is used and uncovers 3 clusters of nodes. ETM edge information is added to the network through the 3 edge colours green, grey and blue. On the right hand side, ETSBM clustering results uncover 4 clusters. The cluster coloured in green, in the middle of the figure, is split into two cluster on the right hand side, the green one and the red one, each discussing of a different topic, the blue and grey topic respectively. The clusters of nodes of the figure on the right-hand side are coherent both in terms of topology and topics of discussion contrary to the figure in the middle.}
	\label{fig:scenario_C}	
\end{figure} 

\paragraph{Organisation of the paper} 

The embedded topics for the stochastic block model (ETSBM) is presented in Section \ref{sec:generative_model}.  The inference and the model selection are presented in Section \ref{sec:inference}. Eventually, the model is evaluated against state of the art algorithms on synthetic data and we present results for a real word example build from tweets during the last French presidential election in Sections \ref{sec:numerical_experiment} and \ref{sec:real_world_data}, respectively. Section \ref{sec:conc_disc} presents some concluding remarks and further work.

\section{The ETSBM Model}\label{sec:generative_model}
\subsection{Background and notations}
In this work, we focus on data represented by a directed graph $\mathcal{G}=\{\mathcal{V}, \mathcal{E}\}$, such that $\mathcal{V} = \{1,\dots, M\}$ denotes the set of nodes and $\mathcal{E} := \{ (i,j) : i,j \in \{1,\dots, M \}, i \rightsquigarrow j \}$ the set of edges, where $ i \rightsquigarrow j $ indicates that $i$ is connected to $j$. The connections, or edges, are represented by a binary matrix $A \in \mathcal{M}_{M \times M}( \{0,1\})$ such that $i$ is connected to $j$, or $(i,j) \in \mathcal{E}$, if and only if $A_{ij}=1$. In the applications we consider, this implies that node $i$ sent textual information to $j$ such as one or a series of emails for instance. These texts are denoted $W_{ij} = \{ W_{ij}^{1}, \dots, W_{ij}^{D_{ij}} \}$ with $D_{ij}$ the number of documents sent from $i$ to $j$ and are gathered in the collection $ W = \left\{W_{ij} , (i,j) \in \mathcal{E} \right\}$. Each document $d$ in $W_{ij}$ is a collection of words of size $N_{ij}^d$, i.e $W_{ij}^d = \{ w_{ij}^{d1},\dots, w_{ij}^{d N_{ij}^{d} } \}$. The size of the vocabulary is denoted $V$ and the words are identified by their index in the vocabulary: each word $w$ is in $\{ 1, \dots, V\}$.
Finally, only graphs without self loops are considered in this work, therefore $A_{ii} =0$ for all $i \in \mathcal{V}$. Notice that all the present work can be extended to undirected networks using $W_{ij} = W_{ji}$ for all pairs $(i,j)$ such that $A_{ij} = A_{ji} = 1$. The directed case is more adequate to messages sent from $i$ to $j$ while the undirected case is better suited for co-authorships networks for instance.

The notation $\mathcal{M}_{d \times p}(\mathbb{F})$ will be used to denote the matrix space with matrix of dimension $d \times p$ and coefficients in $\mathbb{F}$ while the notation $\mathcal{M}_{d}(N, \omega)$ will be used to denote the multinomial distribution with parameters $N \in \mathbb{N}$ and $\omega \in \Delta_{d - 1}$ where 
\begin{align*}
\Delta_{d - 1}~=:~\left\{ x \in \mathbb{R}^d : \forall i \in \{1, \dots,d\},  x_i \geq 0, \sum_{i=1}^d x_i = 1 \right\}.
\end{align*}

\subsection{Modelling the interactions}
 In this work, we assume that each node belongs to a single cluster. Moreover, we assume that the connexion probability between two nodes only depends on the cluster memberships. Indeed, let $Y_i$ denotes the cluster membership of node $i$ for any $i \in \{1, \dots, M \}$. All $Y_i$ are assumed to follow a multinomial distribution and to be independent and identically distributed (\textit{i.i.d}), given the cluster proportions $\gamma \in \Delta_{Q-1}$, lying in the simplex of dimension $Q$,
\begin{align*}
	Y_i \mid \gamma \overset{\text{i.i.d}}{\sim} \mathcal{M}_{Q}(1, \gamma).
\end{align*}
Thus, each node $i$ is associated with cluster $q$ with probability $\gamma_q$. Then, we define the cluster membership matrix $Y$ by stacking the node cluster membership vectors $(Y_i)_i$ together such that $Y = (Y_1 \cdots Y_M)^{\top} \in \mathcal{M}_{M \times Q}(\{0,1\})$. The probability of $Y$ is given by
\begin{align}\label{eq:cluster_memb_proba}
	p(Y \mid \gamma) & =\prod_{i=1}^M \prod_{q=1}^Q \gamma_{q}^{Y_{iq}}.
\end{align}
Besides, the connections between nodes are supposed to be independent given their cluster memberships. Moreover, if nodes $i$ and $j$ are respectively in clusters $q$ and $r$, an edge is assumed to be present with probability $\pi_{qr}$,
\begin{align}
	A_{ij} \mid Y_{iq}Y_{jr} = 1, \pi_{qr}  \overset{\text{i.i.d}}{\sim} \mathcal{B}(\pi_{qr}),
\end{align}
where $\mathcal{B}(\mu)$ denotes the Bernoulli distribution with probability $\mu$. Thus, given the cluster memberships of the nodes $Y$ and the probability matrix $\pi$, the probability of all node connections is given by 
\begin{align}\label{eq:adjacency_proba}
	p(A \mid Y, \pi) & = \prod_{i\neq j}^M \prod_{q,r}^Q \left(\pi_{qr}^{A_{ij} } (1-\pi_{qr})^{(1- A_{ij}) } \right)^{Y_{iq} Y_{jr}}. 
\end{align}
Eventually, the joint-probability of the adjacency matrix $A$, and the cluster memberships vector $Y$,  is obtained by multiplying Equations \eqref{eq:cluster_memb_proba} and \eqref{eq:adjacency_proba},
\begin{align}\label{eq:network_distr}
	p(A, Y \mid \pi, \gamma) & = p(A \mid Y, \pi) 	p(Y \mid \gamma).
\end{align}
Combining Equations \eqref{eq:cluster_memb_proba}, \eqref{eq:adjacency_proba}, and \eqref{eq:network_distr}, we retrieve the SBM distribution \citep{daudin2008mixture}.

\subsection{Modelling the texts}
Our approach extends ETM to capture information of groups of texts.  Essentially, texts are assumed to be generated according to a mixture of topics with latent topic vectors only depending on node clusters. More precisely, a text sent from node $i$ in cluster $q$ to node $j$ in cluster $r$ is assumed to have a logistic-normal topic proportion vector $\theta_{qr} = (\theta_{qr1} , \dots, \theta_{qrK})^{\top} \in \Delta_{K-1}$, with the number of topics $K$ fixed beforehand. It is obtained by applying the softmax function to a Gaussian random vector $\delta_{qr}$,
\begin{align*}
	& \delta_{qr}  \sim \mathcal{N}(0_K, I_{K}), \nonumber\\
	& \theta_{qr}  = \softmax(\delta_{qr}),
\end{align*}
where $\softmax(x)~=~\left(\sum_{k=1}^K e^{x_{k}}\right)^{-1} \left(e^{x_{1}}, \dots, e^{x_{K}} \right)^{\top}$.

In the rest of this paper, the notation $\theta = (\theta_{qr})_{1\leq q,r \leq Q}$ is used to refer to the topic proportions while $\delta = (\delta_{qr})_{1\leq q,r \leq Q}$ will refer to the sampling of the random variable. If two nodes $i$ and $j$ are connected and if they are respectively in cluster $q$ and $r$, the words in document $W_{ij}$ are assumed to be \textit{i.i.d}. Indeed, the $n$-th word of the $d$-th documents is assumed to be distributed according a mixture of topics conditionally on the node clusters,
\begin{align}
	W_{ij}^{dn} \mid Y_{iq} Y_{jr} A_{ij} = 1, \theta_{qr}, \alpha, \rho \sim \mathcal{M}_{V}(1, \theta_{qr}^{\top} \beta ), 
\end{align}
where the matrix $\beta= (\beta_1 \cdots \beta_K)^{\top} \in \mathcal{M}_{K \times V}(\mathbb{R})$ corresponds to the distribution over the vocabulary for each topic such that $\beta_k = \softmax \left( \rho^{\top} \alpha_k \right)$ for any $k \in \{1, \dots, K \}$. The matrix $\rho \in \mathcal{M}_{L \times V}(\mathbb{R})$ corresponds to the matrix of the vocabulary embedded into an $L$-dimensional vector space, and $\alpha=(\alpha_1 \cdots \alpha_K) \in \mathcal{M}_{L \times K}(\mathbb{R})$ the matrix of topics represented into the same vector space.

Therefore, the probability of texts can be computed as follow:
\begin{align}\label{eq:proba_texts}
	p(W \mid Y, A, \theta, \alpha, \rho) & =  \prod_{i\neq j}^M \prod_{d=1}^{D_{ij}} p(W_{ij}^{} \mid Y_i , Y_j, A_{ij}=1, \theta, \alpha, \rho) \nonumber \\
	& = \prod_{i\neq j}^M \prod_{d=1}^{D_{ij}} \prod_{n=1}^{N_{ij}^d} \prod_{q,r}^{Q} \prod_{v=1}^V \Bigl( \sum_{k=1}^K \theta_{qrk} \beta_{kv} \Bigr)^{W_{ij}^{dnv} A_{ij}Y_{iq}Y_{jr}} \nonumber\\
	& = \prod_{q,r}^{Q} \prod_{v=1}^V  \Bigl( \sum_{k=1}^K \theta_{qrk} \beta_{kv} \Bigr)^{ W_{qr}^v}.	
\end{align}
The number of time the word $v$ of the dictionary is used in texts sent from cluster $q$ to cluster $r$ is denoted $W_{qr}^v~=~\sum_{i\neq j}^M \sum_{d=1}^{D_{ij}} \sum_{n=1}^{N_{ij}^d}  W_{ij}^{dnv} A_{ij}Y_{iq}Y_{jr}$. Here, $W_{qr}~=~(W_{qr}^1, \dots, W_{qr}^V)^{\top}~\in~\mathbb{N}^{V}$  shall be designated as meta-document $(q, r)$. Moreover, we shall use the bag of words notations such that for any connected pair of nodes $(i,j) \in \mathcal{E}$, $W_{ij}~=~( W_{ij}^1, \dots,W_{ij}^V)^{\top}~\in~\mathbb{N}^{V}$ with for any $v \in \{1, \dots, V \}$, $W_{ij}^v$ represents the total count of word $v$ for all documents sent from $i$ to $j$. 
The model is represented in Figure \ref{graphical_model}.

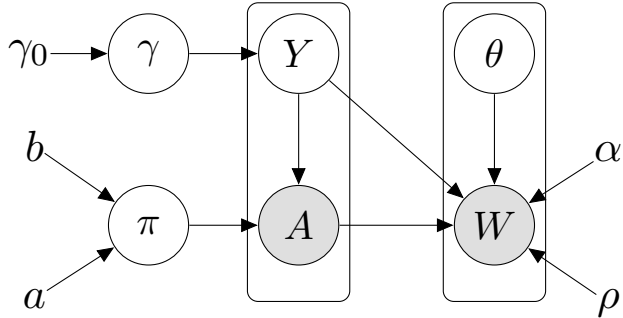
\begin{figure}
	\begin{center}
		\begin{tikzpicture}[scale=1.5, transform shape]

			\node[obs] (W) {$W$};
			\node[latent, above=0.8 of W] (theta) {$\theta$};
			\node[const, below = 0.2 of W, xshift= 1cm] (rho) {$\rho$};
			\node[const, above = 0.2 of W, xshift= 1cm] (alpha) {$\alpha$};
			
			\node[obs, left=1 of W] (A) {$A$};
			\node[latent, above=0.8 of A] (Y) {$Y$};
			\node[latent, left = 0.6 of A] (pi) {$\pi$};
			\node[latent, above = 0.8 of pi] (gamma) {$\gamma$};
			\node[const, below = 0.2 of pi, xshift= -1cm] (a) {$a$};
			\node[const, above = 0.2 of pi, xshift= -1cm] (b) {$b$};
			\node[const, left = 0.50 of gamma] (gamma0) {$\gamma_0$};  
			\edge {pi} {A};
			\edge {gamma} {Y};
			\edge {Y} {A};
			\edge {Y} {W};
			\edge {theta} {W};
			\edge {alpha} {W};
			\edge {rho} {W};
			\edge {A} {W};
			\edge {gamma0} {gamma};
			\edge {a} {pi};
			\edge {b} {pi};
			
			\plate {} { %
				(A)(Y) %
			} {};
			
			\plate {} { (W)(theta) } {};
		\end{tikzpicture}
	\end{center}
	\vspace*{-0.8cm}
	\caption{Graphical representation of the model.}
	\label{graphical_model}
\end{figure}

\subsection{Distribution of the model and links with SBM and ETM.}
Given a cluster configuration $Y$, the joint probability of the model is obtained using Equations \eqref{eq:adjacency_proba} and \eqref{eq:proba_texts} 
\begin{align}\label{eq:joint_distribution}
	p(A, W \mid Y, \alpha, \rho) = p(W \mid Y, A, \alpha, \rho) p(A \mid Y, \pi).
\end{align}

At this point, we emphasise that meta-documents between pairs of clusters of nodes are constructed using the cluster memberships $Y$ and the node connections $A$. Assuming that the cluster membership $Y$ is available as well as all the network information holded by $\pi$ and $\gamma$, the model we propose would simply correspond to ETM applied on the meta-documents $(W_{qr})_{1\leq q,r, \leq Q}$, computed with the available $Y$. 
 
On the other hand, if no texts are exchanged between nodes or the texts are not available, the distribution would reduce to the second term of Equation \ref{eq:joint_distribution}. In that case, the conditional distribution of a standard SBM \citep{daudin2006} is recovered. It is also worth noticing that if a Dirichlet prior is assumed on the topic proportion instead of a logistic-normal, and no factorisation in a embedded latent space is considered, the model corresponds to STBM. By construction, ETSBM generalises SBM and ETM to incorporate both textual data and network information.
	  
\section{Inference} \label{sec:inference}
This section presents the Bayesian framework considered for inference. It also describes the variational-bayes EM algorithm used to maximise the integrated joint likelihood.
\subsection{Bayesian framework for the graph modelling part}
First, a Dirichlet distribution is assumed as a prior distribution on the proportions $\gamma$ of nodes in each cluster,
\begin{align}
	\gamma \sim \mathcal{D}ir_{Q}( \gamma_0 ).
\end{align}
where $\gamma_0$ is set to $(1,\dots, 1) \in \mathbb{R}^Q$, which corresponds to a uniform prior on the simplex.
Moreover, each coefficient of the probability matrix $\pi \in \mathcal{M}_{Q \times Q}(\mathbb{R})$, is assumed to be sampled from from a Beta distribution, such that for any pair $(q,r) \in \{1,\dots,Q\}^2$,
\begin{align*}
	\pi_{qr} \overset{\text{i.i.d}}{\sim} \mathcal{B}eta(a,b).
\end{align*}
In particular, $a$ and $b$ are set to $1$. Thus, the Beta prior corresponds to a Uniform distribution between $0$ and $1$.
 
\subsection{Variational inference}
Eventually, the integrated joint log-likelihood is given by:
\begin{align}\label{eq:log_likelihood}
	\log  p(A, W \mid \alpha, \rho)  = \log \left( \sum_{Y} \int_{\delta} \int_{\gamma} \int_{\pi} p(A, W, Y, \pi, \gamma, \delta \mid \alpha, \rho) d\pi d\delta d\gamma \right).
\end{align}
Unfortunately, this quantity is intractable since it requires computing it for the $Q^M$ configurations of $Y$, which is naturally computationally too demanding. Moreover, the integral with respect to $\delta$ is not tractable either because of the $\softmax$ function. Thus, it cannot be optimised directly. However, it is possible to overcome this issue using a variational-bayes expectation-maximisation algorithm (VBEM) \cite{attias1999variational}. This comes handy as it makes the inference scalable to large datasets.

The variational approach consists in splitting Equation \eqref{eq:log_likelihood} in two terms using a surrogate distribution on $Y, \pi, \gamma$ and $\delta$, denoted $R(Y, \pi, \gamma, \delta)$.

\begin{proposition}\label{prop:elbo_decomposition}
	Denoting $R(\cdot)$, a distribution on $Y, \pi, \gamma$ and $\delta$, the integrated joint log-likelihood can be decomposed as follow: 
	\begin{align*}
		\log  p(A, W \mid \alpha, \rho) & =  \mathscr{L} (R(\cdot); \alpha, \rho ) + \KL(R(\cdot) || p(Y, \pi, \gamma, \delta \mid A, W, \alpha, \rho)),
	\end{align*}
	where
	\begin{align*}
		\mathscr{L} (R(\cdot); \alpha, \rho ) = \sum_{Y} \int_{\pi, \gamma, \theta} R(Y, \pi, \gamma, \delta) \log \frac{ p(A, W, Y, \pi, \gamma, \delta \mid \alpha, \rho) }{ R(Y, \pi, \gamma, \delta)} d\pi d \delta d \gamma.
	\end{align*}
\end{proposition}
\begin{proof}
 The proof is provided in \ref{appendix:elbo_decomp}.
\end{proof}

 To make $\mathscr{L}(R(\cdot); \alpha, \rho)$ tractable, we use the following mean-field assumption :
\begin{align}\label{eq:mean_field_assumption}
R(Y, \pi, \gamma, \delta)=R(Y) R(\pi) R(\gamma) R\left( \delta \right) .
\end{align}
Following the optimality results of \cite{latouche2012variational}, we impose the following variational distributions: 
\begin{align}\label{eq:var_distr}
	R(Y) & = \prod_{i=1}^M R(Y_i)=  \prod_{i=1}^M \mathcal{M}_{Q} (Y_i ; 1, \tau_i),  \nonumber\\
	R(\pi) & = \prod_{q,r = 1}^Q  R(\pi_{qr}) = \prod_{q,r = 1}^Q \mathcal{B}eta(\pi_{qr}; \tilde{\pi}_{qr1}, \tilde{\pi}_{qr2} ), \nonumber \\
	R(\gamma) & = \mathcal{D}ir_{Q}(\gamma; \tilde{\gamma}).
\end{align}
Each vector $\tau_i$ is of size $Q$ and encodes the (approximate) posterior probabilities for node $i$ to be in each cluster.
Given $\tau = (\tau_i)_i$, the set of posterior cluster membership probabilities, for any pair $(q,r)$ the corresponding expected meta-document can be computed as 
\begin{align}\label{eq:expected_meta_docs}
	\tilde{W}_{qr} = \sum_{i\neq j } \tau_{iq} \tau_{jr} W_{ij}.
\end{align}
By construction, the $v$-th element of vector $\tilde{W}_{qr} $ is the expected pseudo count of word $v$ for all documents sent from nodes in cluster $q$ to nodes in cluster $r$. 
Finally, the variational distribution on latent topic proportions is assumed to be:
\begin{align}\label{eq:variational_distr_delta}
	R(\delta) = \prod_{q,r = 1}^Q R(\delta_{qr}) =  \prod_{q,r = 1}^Q \mathcal{N}\bigl(\delta_{qr}; \mu_{qr}(\tau, \nu), \diag(\sigma_{qr}^2(\tau, \nu)) \bigr),
\end{align}
with $(\mu_{qr}(\tau, \nu), \sigma_{qr}(\tau, \nu) )^{\top} = f(\tilde{W}_{qr}^{norm}(\tau) ; \nu)$ the output of a parametric function, typically a (deep) neural network, with parameters denoted $\nu$. Hereafter, the ETM encoder will be used as the function $f$  parametrised by $\nu$. The normalised expected meta-documents  $\tilde{W}_{qr}^{norm}(\tau) = \left(\sum_{v=1}^V \tilde{W}_{qr}^v (\tau)\right)^{-1} \tilde{W}_{qr}(\tau) \in \mathbb{R}^V$ are then given to the encoder which outputs the mean and variance vectors  $(\mu_{qr}(\tau, \nu), \sigma_{qr}(\tau, \nu) )^{\top}$ of the posterior distribution. Our inference strategy is inspired by \cite{dieng2020topic} and finds its roots in the original work of  \cite{kingma2014autoencoding} for classical data. However, as we shall see, a critical property of our methodology is that the (approximate) posterior allocation probabilities $\tau$ will change through the updates and so are the inputs of the encoder. In all experiments we carried out, we used a 3-layer architecture with 800 units for the hidden layers, as originally proposed in \cite{dieng2020topic}. In order not to increase the number of parameters $\nu$ linearly with the number of pairs of groups, amortised inference is used as advocated in \cite{gershman2014amortized} or \cite{kingma2014autoencoding}.

\begin{proposition}\label{prop:elbo_parametric_function}
	Using the assumptions describes in Equations \eqref{eq:mean_field_assumption}, \eqref{eq:var_distr} and \eqref{eq:variational_distr_delta}, the ELBO, which is a functional of the variational distribution, reduces to a function of the variational parameters and can be split in two terms associated with the network and with the texts respectively:
	\begin{align}
	\mathscr{L}(R(\cdot); \alpha, \rho ) & = \mathscr{L}( \tau, \tilde{\pi}_{1}, \tilde{\pi}_{2}, \tilde{\gamma}, \nu ; \alpha, \rho ) \\
	& = \mathscr{L}^{net}(\tau, \tilde{\pi}_{1}, \tilde{\pi}_{2}, \tilde{\gamma} ; \alpha, \rho ) + \mathscr{L}^{texts}(\tau, \nu ; \alpha, \rho ),	
\end{align}
where $\tilde{\pi}_{1} =(\tilde{\pi}_{qr1})_{qr}$, $ \tilde{\pi}_{2} = (\tilde{\pi}_{qr2})_{qr}$.
\end{proposition}
\begin{proof}
	The proof and the exact value of the ELBO is detailed in \ref{prop:elbo_detail_proof}
\end{proof}

\subsection{Optimisation and Algorithm}

We now aim at maximising the ELBO with respect to the variational parameters $\tilde{\pi}, \tilde{\gamma}, \tau$ and $\nu$ and to the parameters $\rho$ and $\alpha$. On the one hand, following \cite{latouche2012variational}, the variational parameters $\tilde{\pi}$ and $\tilde{\gamma}$ only depend on $\tau$ and are updated as follow:
\begin{align}
	\tilde{\gamma_{q}} & = \gamma_{0q} + \sum_{i=1}^M \tau_{iq}  \nonumber \\
	\tilde{\pi}_{qr1} & = \pi^0_{qr1} + \sum_{i \neq j}^M \tau_{iq} \tau_{jr} X_{ij}, \ \ 
	\tilde{\pi}_{qr2} = \pi^0_{qr2} + \sum_{i \neq j}^M \tau_{iq} \tau_{jr} (1-X_{ij}).
\end{align} 

On the other hand, $\nu$, as well as $\rho$ and $\alpha$ are optimised by a stochastic gradient descent algorithm using Pytorch automatic differentiation \citep{pytorch2019nips} and the Adam optimiser \citep{kingma2014adam} with a learning rate of $10^{-4}$. Once both parts are done, we only need to update $\tau$ using the already up-to-date parameters. To do so, we switch from $\tau$ lying on the simplex $\Delta_{Q-1}$ to the unconstrained space $\mathbb{R}^{Q-1}$ using for any $i \in \mathcal{V}$ and $ q \in \{1, \dots, Q-1\}$:
\begin{align*}
	\xi_{iq} & = \ln(\tau_{iq}) - \ln(\tau_{iQ}).
\end{align*}
We then use the automatic differentiation and the Adam optimiser with a learning rate of $0.55$ to maximise the ELBO with respect to $\xi$. It is worth emphasising that the ELBO is optimised over the whole set of allocation probability vectors $\tau = (\tau_i)_i$ contrary to STBM which looks for a hard allocation of nodes to clusters, one allocation being optimised at a time, all the others being fixed. Moreover, by optimising the entry of the encoder through  $\tau$, thus looking for an optimal allocation of documents to pairs of clusters, the moves in $\tau$ aim at uncovering the optimal direction in the posterior distribution in $(\theta_{qr})_{qr}$ maximising the ELBO. In that regard, ETSBM has links with the quasi branching bound algorithm of \cite{jouvin2021greedy}  for document clustering. Considering a unique core for illustration, on an Intel(R) Core(TM) i7-10875H 2.30 GHz CPU and a Nvidia GeForce RTX 2080 Super 8 Go GPU, it takes about 15 seconds to analyse a dataset with 100 nodes and 1 000 documents. Moreover, studying a dataset with more than 200 000 documents, characterising all the connections between 1500 nodes,  is done in approximately 6 minutes. In practice, we emphasise that the running time can be reduced even more by considering extensive parallelisation as well as stochastic variational inference strategies adapted for networks as in \cite{gopalan2013efficient}. The Python implementation of the complete methodology we propose is available at \href{https://plmlab.math.cnrs.fr/rboutin/etsbm_package}{https://plmlab.math.cnrs.fr/rboutin/etsbm\_package}.

\subsection{Model selection}\label{sec:model_selection}

Finally, the selection of the number of cluster $Q$ is performed using the ELBO. It is useful to remind that the aim of the model is to select the number of clusters providing the more meaning. Therefore, relying on \cite{latouche2012variational}, we take advantage of the Bayesian framework that automatically penalises the complexity of the model with respect to $Q$. The best number of cluster $Q$ is then selected by estimating the parameters for models with different number of cluster $Q$ and keeping the one with the highest ELBO. Our experiment Section \ref{table:selec_mod_sc_C} confirms that this procedure provides a relevant model selection criterion. In this paper, the number of topics $K$ is not selected. Indeed, we choose to keep a high $K$ as advocated in \cite{dieng2020topic}. In practice, once the inference of the topics is done, a classical approach consists in focusing  the interpretation on the results associated with the most frequent topics. As we shall see, in the experiment section, provided that the value of $K$ chosen is large enough, the proposed procedure provides an accurate estimate of $Q$.

\section{Numerical experiments} \label{sec:numerical_experiment}
In this section, a series of experiments is presented to assess the proposed methodology. First, three scenarii used for benchmarking are described. Second, an illustration of the results provided by ETSBM on a simulated dataset from one of the scenarii is given. Then, results from experiments to evaluate the model selection criterion on the three scenarii considered are brought. Moreover, various strategies to initialise ETSBM are compared. Finally, an extensive set of experiments on the three scenarii with three levels of difficulty is carried out to evaluate the clustering performances of ETSBM against competitive algorithms.

\subsection{Simulation setup}\label{sec:simulation_setup}

The networks with textual edges are generated following three scenarii $A$, $B$, $C$, as originally introduced in  \cite{bouveyron2018stochastic}.

\paragraph{Sampling networks with textual edges}\label{seq:simulation_scenarii}
\begin{itemize}
	\item Scenario $A$ is composed of three communities, each defining a cluster, and four topics. By definition, a community is defined such that more connections are present between nodes of the same community. For each cluster, a specific topic is employed to sample all the documents associated with the corresponding intra-cluster connections. Besides, an extra topic is considered to model documents exchanged between nodes from different clusters. Thus, by construction, the clustering structure can be retrieved either using the network or the texts only.
	\item Scenario $B$ is made of a single community and three topics. Thus, all nodes connect with the same probability. Then, the community is split into two clusters with their respective topics. An extra topic is used to model documents exchanged between the two clusters. Therefore, in such a scenario, the network itself is not sufficient to find the two clusters but the documents are.
	\item Scenario $C$ is composed of three communities and three topics. Two of the communities are associated with their respective topics, say $t_1$ and $t_2$. Moreover, following the previous scenario, the third community is split in two clusters, one being associated with topic $t_1$ and the other with $t_2$. Thus, considering both texts and topology, each network is actually made of four node clusters. Fundamentally, both textual data and the network itself are necessary to uncover the clusters. This scenario will be of major interest in this experiment section since it allows to ensure that the two sources of information are correctly used to retrieve partitions.
\end{itemize}
The edges holding the documents are constructed by sampling words from four BBC articles, focusing each on a given topic. The first topic deals with the UK monarchy, the second with cancer treatments, and the third with the political landscape in the UK. The last topic deals with astronomy. In the general setting, for all scenarii, the average text length for the documents is set to 150 words.
 The parameters used to sample data from the three scenarii are given in Table \ref{tab:summary_scenarii}. Moreover,
three examples of networks generated from $A$, $B$ and $C$ are presented in Figure \ref{fig:example_sc}.

\paragraph{Clustering performance evaluation} The main criterion used in the following to evaluate the clustering performances of the different strategies is the adjusted random index (ARI). ARI measures how close two partitions are from one another. The closer ARI is to 1, the better the results are. A random cluster assignment leads to an ARI of 0, while a perfect retrieval of the cluster memberships gives an ARI of 1.

\hspace{-2cm}\begin{table}[ht]
	\begin{tabular}{|>{\centering}p{0.4\textwidth}|ccc|}
		\hline
		 & Scenario $A$ & Scenario $B$ & Scenario $C$\\
		\hline
		$Q$ (clusters) & 3  & 2  & 4 \\
		\hline
		$K$ (topics) & 4 & 3 & 3 \\
		\hline
		Communities & 3 & 1 & 3 \\ 
		\hline
		{$\pi_{qr}$ (connection probabilities)\\ $\eta=0.25$, $\epsilon=0.01$ } & $ 
		\begin{pmatrix}
			\eta & \epsilon  & \epsilon \\
			\epsilon & \eta & \epsilon  \\
			\epsilon & \epsilon & \eta
		\end{pmatrix} $ & 
		 $\begin{pmatrix}
			\eta & \eta  \\
			\eta & \eta  
		\end{pmatrix} $ &
		$ 
		\begin{pmatrix}
			\eta & \epsilon  & \epsilon & \epsilon \\
			\epsilon & \eta & \epsilon & \epsilon  \\
			\epsilon & \epsilon & \eta & \eta \\
			\epsilon & \epsilon & \eta & \eta
		\end{pmatrix} $
	    \\
		\hline
		Topics between pairs of clusters $(q, r)$ & $ 
		\begin{pmatrix}
			t_1 & t_4  & t_4 \\
			t_4 & t_2 & t_4  \\
			t_4 & t_4 & t_3
		\end{pmatrix} $ & 
		$\begin{pmatrix}
			t_1 & t_3  \\
			t_3 & t_2  
		\end{pmatrix} $ &
		$ 
		\begin{pmatrix}
			t_1 & t_3  & t_3 & t_3 \\
			t_3 & t_2 & t_3 & t_3  \\
			t_3 & t_3 & t_1 & t_3 \\
			t_3 & t_3 & t_3 & t_2
		\end{pmatrix} $
		\\ 
		\hline
		Sufficient information to uncover the clusters & Network & Topics & Network \& Topics\\
		\hline
	\end{tabular}
\caption{Detail of the three simulation scenarii to evaluate our model.}
\label{tab:summary_scenarii}
\end{table}

\begin{figure}
	\begin{subfigure}{0.32\textwidth}
		\includegraphics[width=\textwidth]{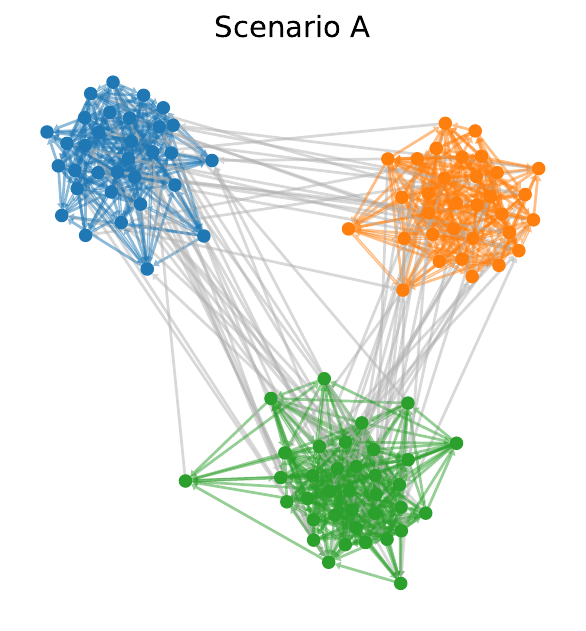}
	\end{subfigure}
	\hfill
	\begin{subfigure}{0.32\textwidth}
		\includegraphics[width=\textwidth]{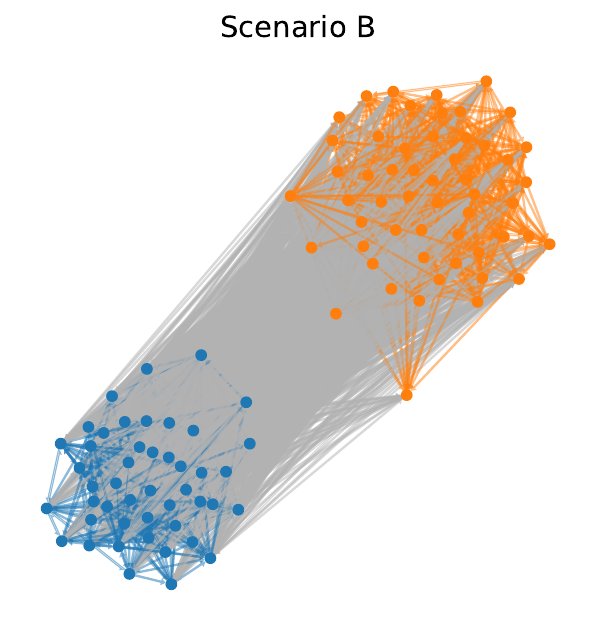}
	\end{subfigure}
	\hfill
	\begin{subfigure}{0.32\textwidth}
		\includegraphics[width=\textwidth]{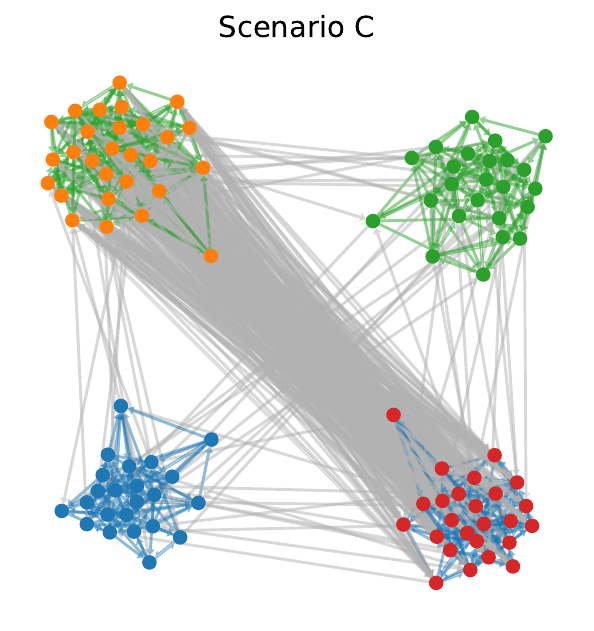}
	\end{subfigure}
	\caption{An example of each scenario is presented. The node colours denote the cluster memberships and the edge colours denote the most-used topic within the corresponding documents. The Scenarii $A$, $B$ and $C$ are composed of 3, 1 and 3 communities respectively.}
	\label{fig:example_sc}
\end{figure}

\paragraph{Different levels of difficulties}\label{level_diff}

To evaluate ETSBM against state of the art STBM in Sections \ref{seq:initialization} and \ref{seq:benchmark}, two levels of difficulty are introduced. The first one, named \textit{Hard 1}, makes it particularly hard to distinguish connectivity patterns by using an intra-cluster connectivity probability of 0.2. In Table \ref{tab:summary_scenarii}, it corresponds to $\epsilon = 0.2$ instead of $0.01$. The second one, named \textit{Hard 2}, introduces difficulty on the text part by using smaller texts of $110$ words on average instead of $150$ and by adding noise. 
 In our case, this translates into fixing:
\begin{align}
	\theta_{qr} = (1-\zeta) \theta_{qr}^{\star} + \zeta * \left(\frac{1}{K}, \dots, \frac{1}{K}\right)^{\top},
\end{align}
with $\zeta = 0.7$. Thus, for each pair of clusters $(q,r)$, the texts are sampled according to a mixture between a multinomial distribution with probability 1 on the corresponding topic and a uniform distribution over all topics considered. Finally, the intra-cluster connection probability is decreased from $0.2$ to $\eta = 0.1$.

\subsection{An introductory example}
A first glimpse at the ETSBM results on a single network simulated with Scenario $C$ is presented here. In Figure \ref{fig:sc_C_elbo_ari}, the evolution of the ELBO and ARI values are monitored at each iteration of the inference of ETSM applied on this single simulated network.  
\begin{figure}[ht]
\centering
\includegraphics[width=\textwidth]{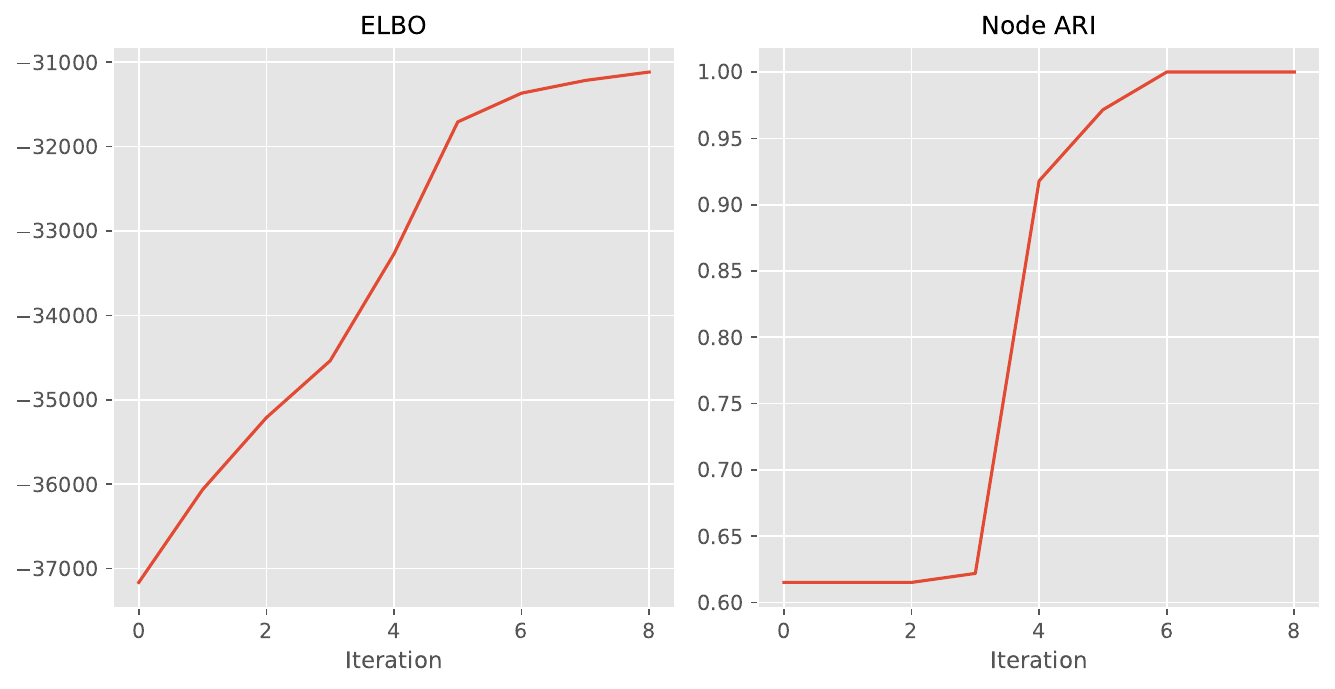}
\caption{Evolution of ETSBM ELBO and ARI (y-axis) at each iteration (x-axis) on the Scenario $C$ after an initialisation with the K-means algorithm.}
\label{fig:sc_C_elbo_ari}
\end{figure}
As we can see, both the ELBO and the ARI increase after each iteration. In particular, starting from the clustering initialisation with an ARI value of $0.62$, the algorithm converges to a value of $1$, characterising a perfect cluster recovery. This figure illustrates the ability of the methodology proposed to retrieve the true node partition, by combining the textual and network data.
 
In addition, Figure \ref{fig:sc_C_pi_gamma} provides representations for the expected posterior estimates $\hat{\pi}$ and $\hat{\gamma}$ computed as follows $\hat{\pi}_{qr} =  \tilde{\pi}_{qr1} / (\tilde{\pi}_{qr1} + \tilde{\pi}_{qr2})$ and $\hat{\gamma}_q = \tilde{\gamma}_q / (\sum_{r=1}^Q \tilde{\gamma}_r)$. We emphasise that the matrix characterises the connexion probabilities between clusters with a $10^{-2}$ rounding. It matches the expected connectivity structure described in Table \ref{tab:summary_scenarii}. 

\begin{figure}[ht]
	\centering
	\includegraphics[width=\textwidth]{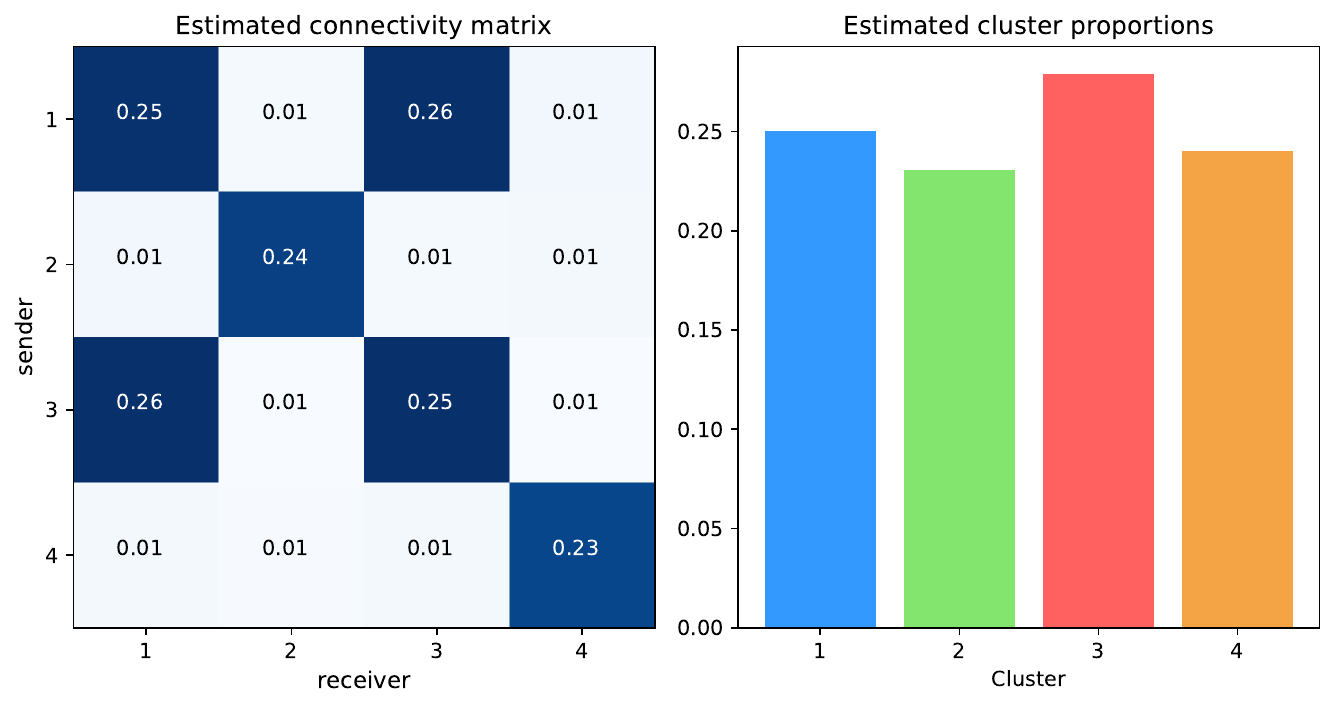}
	\caption{On the left hand side, the expected posterior estimate of the connectivity matrix $\pi$ provided by ETSBM. On the right hand side, the expected posterior estimate of the cluster proportions $\gamma$. The graph was generated following Scenario $C$.}
	\label{fig:sc_C_pi_gamma}
\end{figure}

Eventually, the topics learnt as well as the clustering results on the network are presented in Figure \ref{fig:sc_C_topics_network}. In the network representation, the node colours correspond to the cluster memberships while the edge colours indicate the most used topic in the corresponding documents. Moreover, for each topic $t_k$ with $k \in \{1, 2, 3\}$, the 10 words with the highest probabilities, according to the corresponding topic vector $\beta_k$, are displayed. The three topics presented are well-separated and can be identified as the topics dealing respectively with astronomy, the political landscape in the UK, and the UK monarchy, as expected. In addition, four node clusters have been retrieved and the edge topics, or colours, match the description of the Scenario $C$ setup. To conclude, ETSBM successfully render both the network topology and the edge topics.
\begin{figure}[ht]
	\centering
	\includegraphics[width=\textwidth]{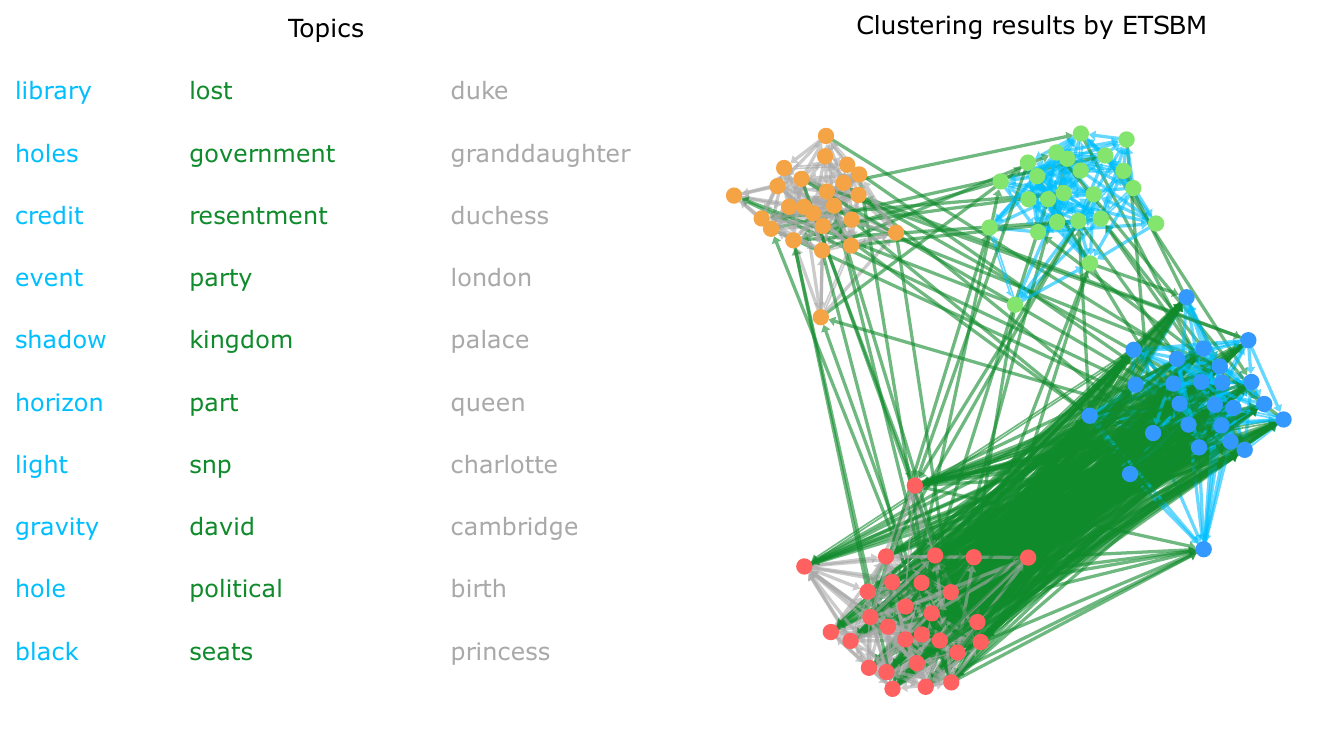}
	\caption{On the left hand side, the top 10 words of each topic according to ETSBM results. Thus, for each topic $t_k$ with $k \in \{1, 2, 3\}$, the 10 words with the highest probability values, according to the corresponding topic vector $\beta_k$, are displayed. On the right hand side, ETSBM clustering result is illustrated. The node colours indicate the node clusters while the edge colours correspond to the most used topic within the document.}
	\label{fig:sc_C_topics_network}
\end{figure}

Finally, Figure \ref{fig:sc_C_meta_graph} provides a high level representation of the results. On the one hand, the ``meta-nodes'' represent  ETSBM clusters and their size is proportional to the number of nodes assigned to the corresponding clusters. Moreover, the ``meta-node'' colours are consistent with the colours in Figure \ref{fig:sc_C_topics_network}. On the other hand, the edges represent the meta-documents. We recall that they correspond to the expected posterior estimate of a document for a given pair of clusters. The edge colours correspond to the most used topic within the meta-document. The edge widths are determined by the posterior probabilities of connections between pairs of clusters. This figure underlines ETSBM capability to produce intelligible and accurate data summary. We emphasise that graphs with thousands of edges, that sometimes cannot be represented because of memory issues, are here able to be summarised in easy-to-read meta-graphs. 
 
\begin{figure}[ht]
	\centering
	\includegraphics[width=0.6\textwidth]{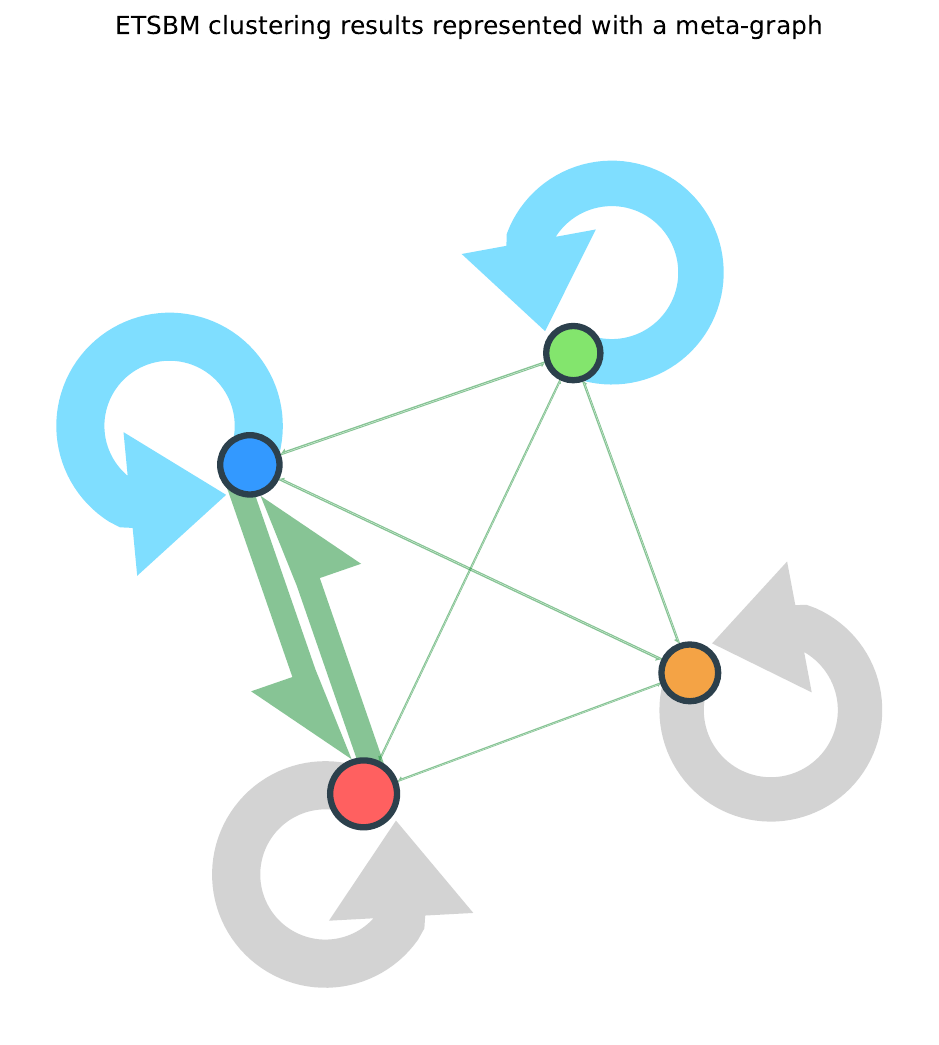}
	\caption{Meta representation of ETSBM results. On the one hand, the clusters are represented by the node colours, the node widths are proportional to the expected posterior estimate of the cluster proportions, and their colours correspond to the same cluster colours as in the network in Figure \ref{fig:sc_C_topics_network}. On the other hand, the edges are coloured as the most used topic within the meta-document and the widths are proportional to the posterior probabilities of connections between clusters. }
	\label{fig:sc_C_meta_graph}
\end{figure}

To conclude, this introductory example showed the ETSBM capacity to render meaningful summaries by combining both network and text information. It is worth reminding that, since it comes from Scenario $C$, those results could not have been retrieved with models handling only network or texts as SBM, LDA or ETM.

\subsection{Effect of the initialisation}\label{seq:initialization}
This experiment aims to evaluate the impact of the initialisation on the final performance of our methodology. The networks are generated according to the \textit{Hard 2} difficulty, to easily visualise the differences between the tested configurations. Moreover, the experiment is performed on Scenario $C$ to ensure both the network and textual data are used. Three different initialisations are compared:  clusters may be randomly assigned to the nodes (random), or initial clusters can be determined by a K-Means algorithm fitted on the adjacency matrix $A$. Finally, the dissimilarity procedure proposed in \cite{bouveyron2018stochastic} is evaluated as the last initialisation strategy (dissimilarity). It uses both network and textual information to build a similarity matrix based on the topics discussed between nodes. Then, a K-means algorithm is performed on this similarity matrix to find a cluster allocation for each node. This initialisation strategy requires to provide the topic proportion of each edge. Thus, ETM is trained on the texts and the estimated topic proportions $(\theta_{ij})_{(i, j) \in \mathcal{E}}$ are used for the dissimilarity initialisation. Figure \ref{fig:init_impact} presents the ARI results with, for each initialisation strategy, a boxplot of the raw initialisation and of ETSBM clustering.
\begin{figure}[ht]
	\centering
	\includegraphics[width=0.9\textwidth]{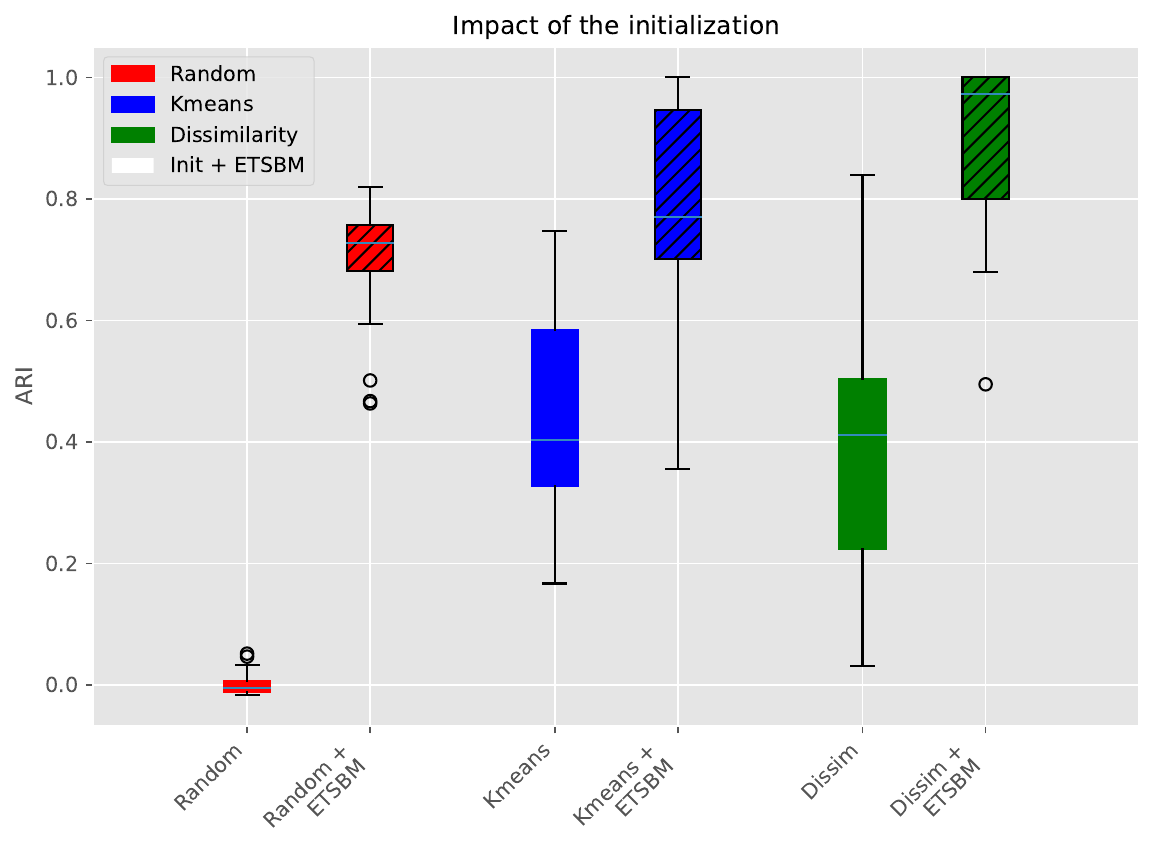}
	\caption{This figure displays the boxplots of the initialisation ARI (the boxplot without stripe) and of ETSBM clustering ARI with the same initialisation (the boxplot with stripes). This experiment was performed on 50 networks generated following Scenario $C$ in the \textit{Hard 2} setting.}
	\label{fig:init_impact}
\end{figure}

 While the random initialisation is close to 0 for ARI, both the K-means and the dissimilarity initialisation fluctuates in terms of ARI, with no clear advantage for one of the two strategies. However, ETSBM provides much better results with the dissimilarity initialisation than with K-means. It is also worth noticing that the gap between the random and K-Means initialisations has largely been closed by ETSBM algorithm. One possibility is that the model suffers the same flaws as SBM, which is for the ELBO to fall into local minimum. It is possible that the use of texts in the dissimilarity limits this effect. Therefore, we will only use the dissimilarity initialisation in the rest of the paper as it provides the best results in most cases.

\subsection{Model selection}
This experiment aims to assess the efficiency of the model selection criterion, presented in Section \ref{sec:model_selection}. Let us remind that we do not aim at selecting the number of topics $K$ since it is handled afterwards. As a consequence, the model selection criterion is evaluated for different values of $K$ to ensure that the performances remain high, in all cases. For each scenario, 50 networks are sampled following the setup described in Section \ref{seq:simulation_scenarii}. For each network, ETSBM parameters are estimated taking the best initialisation out of 10. Table \ref{table:selec_mod_sc_C} presents the percentage of time a number $Q$ is selected using the strategy proposed in  Section \ref{sec:model_selection} over the 50 networks, for each $K$ value. It is worth noticing that the right model is selected more than $75\%$ of the time, except for the Scenario $B$ with $K=5$, slightly bellow with $68\%$. In addition, as advocated before, for $K=10$, the right model is selected more than $80\%$ of the time in each scenario. This experiment illustrates the capacity of the model selection criterion to retrieve the number of clusters. Moreover, keeping a high value of $K$ is confirmed to be compatible with an efficient cluster number selection. 

\begin{table}
	\begin{tabular}{lrrrrr@{\hskip 0.5in}rrrrr@{\hskip 0.5in}rrrrr}
		\toprule
		\multirow{2}{*}{\diagbox{K}{Q}} & \multicolumn{5}{c}{Scenario $A$} & \multicolumn{5}{c}{Scenario $B$} & \multicolumn{5}{c}{Scenario $C$}\\
		&  2  &  $\mathbf{3}$  &  4  &  5  &  10 & $\mathbf{2}$  &  3  &  4  &  5  &  10 & 2  &  3  &  $\mathbf{4}$  &  5  &  10 \\
		\midrule
		2  &   0 &  \textbf{94} &   6 &   0 &   0 &  \textbf{74} &  24 &   2 &   0 &   0 &   0 &   0 &  \textbf{92} &   8 &   0\\
		3  &   0 & \textbf{ 90} &  10 &   0 &   0 &  \textbf{78} &  18 &   4 &   0 &  0 &   0 &   0 &  \textbf{90} &  10 &   0 \\
		4  &   0 &  \textbf{78} &  20 &   2 &   0 &  \textbf{76} &  20 &   4 &   0 &   0&   0 &   0 &  \textbf{94} &   6 &   0\\
		5  &   0 &  \textbf{86} &  14 &   0 &   0 &  \textbf{68} &  28 &   4 &   0 &   0 &   0 &   0 &  \textbf{84} &  16&   0\\
		10 &   0 &  \textbf{88} &  10 &   2 &   0 &  \textbf{82} &  18 &   0 &   0 &   0 &   0 &   0 &  \textbf{86} &  14&   0\\
		\bottomrule
	\end{tabular}
	\caption{This table presents the percentage of time a number of clusters have been selected on 50 simulated networks. The experiment is repeated for different values of $K$, and for Scenario $A$, $B$ and $C$. For instance, in Scenario $A$ with $K=3$, the model with $Q=3$ clusters was selected in 90$\%$ of cases. }
	\label{table:selec_mod_sc_C}
\end{table}

\subsection{Benchmark study}\label{seq:benchmark}
To end this section, ETSBM is evaluated against state of the art clustering algorithms for STBM.
We recall that STBM is currently the only algorithm capable of simultaneously analysing the texts on the edges as well as the node connections to cluster the nodes. In order to provide baselines, we also give the results obtained with SBM as well as a spectral clustering algorithm (SC) presented in \cite{shi2000normalized, von2007tutorial}, with a radial basis function as a kernel and a normalised symmetric Lagrangian. Those methods are evaluated on the three levels of difficulty presented in Section \ref{level_diff}.  Besides, results for LDA as well as ETM for text clustering are also provided. For each level of difficulty and each scenario, Table \ref{tab:benchmark} displays the mean and the standard deviation of the ARI values obtained over 50 graphs. Both the node and edge clusters ARI are provided but we recall that the main interest of this work concerns the node clustering performances. In the \textit{Easy} and \textit{Hard 1} settings, the ARI is always 1, which indicates that the true partitions are successfully retrieved by ETSBM and STBM. On the contrarty, SBM and SC are not able to distinguish clusters in Scenario $B$ since all nodes connect one another with the same probability. Identically, in Scenario $C$, SBM and SC alone cannot differentiate the nodes highly connected but discussing of different topics. For instance, in the \textit{Easy} case, this translates into an ARI of $0.01$ and $0.69$ respectively for SBM, and $0.00$ and $0.63$ respectively for SC. In the \textit{Hard 2} setting, ETSBM node clustering significantly outperforms STBM. In particular in Scenario $C$, \textit{Hard 2}, ETSBM results reach an ARI of 0.91 against 0.63 for STBM. Even though it is not the main focus of this work, the edge ARI is always higher than $0.84$, which is satisfactory, and is competitive when not higher than STBM. These significant gaps in the noisy settings highlight ETSBM clustering improvement upon STBM. To conclude, our experiments strongly indicates that ETSBM node clustering performances are either the same or significantly better than STBM.

\begin{table}[t]
\begin{adjustwidth}{-1cm}{}
	\small
	\ra{1.3}
	\centering
	\caption{Benchmark of our model against STBM, SBM, SC and LDA.  When a model does not provide an information, a line is displayed instead of the result. For instance, SBM does not provides edge information.}
	\label{tab:benchmark}
	\begin{tabular}{llcccccc}
		\toprule
		&     & \multicolumn{2}{c}{Scenario $A$} & \multicolumn{2}{c}{Scenario $B$} & \multicolumn{2}{c}{Scenario $C$} \\
		&     &          Node ARI &          Edge ARI &          Node ARI &          Edge ARI &          Node ARI &          Edge ARI \\
		\midrule
		\multirow{5}{*}{\rotatebox{90}{Easy}} & ETSBM &  \textbf{1.00 $\pm$ 0.00}  & \textbf{0.99 $\pm$ 0.03}  &  \textbf{1.00 $\pm$ 0.00}  & \textbf{ 1.00 $\pm$ 0.00}  &  \textbf{1.00 $\pm$ 0.00}  &  \textbf{1.00 $\pm$ 0.00}  \\
		& STBM &  0.98 $\pm$ 0.04  &  0.98 $\pm$ 0.04  &  \textbf{1.00 $\pm$ 0.00}  &  \textbf{1.00 $\pm$ 0.00}  &  \textbf{1.00 $\pm$ 0.00}  &  \textbf{1.00 $\pm$ 0.00}  \\
		& SBM &  1.00 $\pm$ 0.00  &    ------  &  0.01 $\pm$ 0.01  &    ------  &  0.69 $\pm$ 0.07  &    ------  \\
		
	   & SC &  0.97 $\pm$ 0.07  &    ------  &  0.00 $\pm$ 0.01  &    ------  &  0.63 $\pm$ 0.11  &   ------  \\
		
		& LDA &    ------  &  0.97 $\pm$ 0.06  &    ------  &  1.00 $\pm$ 0.00  &    ------  &  1.00 $\pm$ 0.00   \\ 
		& ETM & ------ & 0.96 $\pm$ 0.14 &  ------ & 1.00 $\pm$ 0.00 & ------ & 1.00 $\pm$ 0.00 \vspace{7pt} \\
		
		\multirow{5}{*}{\rotatebox{90}{Hard 1}} & ETSBM &  \textbf{1.00 $\pm$ 0.00}  &  \textbf{0.95 $\pm$ 0.03}  & \textbf{ 1.00 $\pm$ 0.00}  & \textbf{ 1.00 $\pm$ 0.00}  &  \textbf{1.00 $\pm$ 0.00}  &  0.97 $\pm$ 0.04  \\
		& STBM & \textbf{ 1.00 $\pm$ 0.00}  &  0.90 $\pm$ 0.13  & \textbf{ 1.00 $\pm$ 0.00}  & \textbf{ 1.00 $\pm$ 0.00}  &  \textbf{1.00 $\pm$ 0.00}  &  \textbf{0.98 $\pm$ 0.03}  \\
		& SBM &  0.01 $\pm$ 0.01  &    ------  &  0.01 $\pm$ 0.01  &    ------  &  0.01 $\pm$ 0.01  &    ------  \\
		
		& SC &  0.00 $\pm$ 0.02  &    ------  &  -0.00 $\pm$ 0.01  &    ------  &  -0.00 $\pm$ 0.01  &    ------  \\
		
		& LDA &    ------  &  0.90 $\pm$ 0.17  &    ------  &  1.00 $\pm$ 0.00  &    ------  &  0.99 $\pm$ 0.01   \\ 
		& ETM & ------ & 0.93 $\pm$ 0.07 &  ------ & 1.00 $\pm$ 0.00 & ------ & 0.98 $\pm$ 0.03 \vspace{7pt} \\
		\multirow{5}{*}{\rotatebox{90}{Hard 2}} & ETSBM &\textbf{  0.98 $\pm$ 0.06}  &  \textbf{0.83 $\pm$ 0.07}  &  \textbf{1.00 $\pm$ 0.00}  &  0.86 $\pm$ 0.03  &  \textbf{0.91 $\pm$ 0.12}  & \textbf{0.84 $\pm$ 0.12}  \\
		& STBM &  0.75 $\pm$ 0.27  &  0.82 $\pm$ 0.22  &  \textbf{1.00 $\pm$ 0.00}  & \textbf{ 1.00 $\pm$ 0.00}  &  0.63 $\pm$ 0.19  &  0.77 $\pm$ 0.15  \\
		& SBM &  0.96 $\pm$ 0.05  &    ------  &  0.00 $\pm$ 0.00  &    ------  &  0.63 $\pm$ 0.11  &    ------  \\
		
		& SC &  0.98 $\pm$ 0.08  &    ------ &  -0.00 $\pm$ 0.01  &  ------  &  0.60 $\pm$ 0.11  &    ------  \\
		
		& LDA &    ------  &  0.77 $\pm$ 0.09  &    ------  &  0.88 $\pm$ 0.02  &    ------  &  0.84 $\pm$ 0.04  \\
		& ETM & ------ & 0.83 $\pm$ 0.08 &  ------ & 0.85 $\pm$ 0.03 & ------ & 0.86 $\pm$ 0.04\\
		\bottomrule
	\end{tabular}
\end{adjustwidth}
\end{table}

\section{Real World example: analysing the French presidential election with a Twitter dataset}\label{sec:real_world_data}

In this section, we now consider the analysis of a real dataset. We start by describing the context of the study. The dataset is then presented and the results obtained with ETSBM are given. To complete this study, the results obtained with SBM and ETM employed independently are also provided. Finally, a comparison of these results with the ones obtained with ETSBM is performed.
\subsection{Context}
 This section presents a use case on a Twitter dataset dealing with the French presidential election of 2022. The election resulted in Emmanuel Macron being re-elected as President of France. The objective is to use ETSBM to capture the global trends on Twitter before the first round of the French presidential election in April 2022. The network has been constructed using tweets collected by the Linkfluence, a Meltwater company, during a collaboration between journalists of the French newspaper \textit{Le Monde} and two authors of this article \citep{laurent2022gauche}. Newspapers such as \textit{Le Monde} may be interested in having a good understanding of the global dynamics on social media during an electoral period, in order to understand the interest of the public opinion. Thus, interpretable topics and meaningful clusters may help them getting a grasp on the core factors interesting the elector. During the last 50 years, French political landscape has been split between two main parties, the left-democrat, mainly represented by the socialist party, and the right-liberal, represented by \textit{Les Républicains} (formerly UMP). A shift occurred in 2017 when a three-way split between the far-left political families, the centrists, or liberals, and the far-right emerged. This  analysis aims at capturing the major topics discussed prior to the election. In addition, we want to understand the way those topics shape user groups interactions. However, this study does not aim at making any form of prediction about the election.
	
\subsection{Dataset construction and method}
 In the collected data, each node represent a Twitter account. An account $i$ is connected to $j$ if the former retweeted the later or if $i$ ``mentioned'' $j$ with an ``@account$\_$name" in a tweet. The text on the edges are the tweet themselves. Our database has been created by saving any tweet talking about one of the twelve candidates. If several tweets appear from $i$ to $j$, the edge $(i, j)$ holds all those tweets stack together. We only keep edges with text length greater than 100 characters. Then, a lemmatisation procedure is used to reduce the vocabulary size. The ``stopwords'', defined as non-informative words such as ``and'' or ``it'', are withdrawn, as well as numeric characters and words with a length inferior to 3 characters. In the end, we keep the largest connected component of this graph. Our dataset holds $2,730$ nodes and $403,768$ edges. This means that the graph is sparse at $94.58 \%$. We emphasise that this level of sparsity is quite high and makes the data analysis particularly challenging. The number of topics is set to $K=20$. Also, for each $Q$ value, the model is trained for 10 different initialisations and the best result among those 10, ELBO wise, is kept. Then, the number of clusters is selected using our model selection criterion. Figure \ref{fig:real_data_selection_model} shows that the most appropriate model according to our criterion corresponds to a number of clusters $Q = 5$.
\begin{figure}[ht]
	\centering
	\includegraphics[width=0.7 \textwidth]{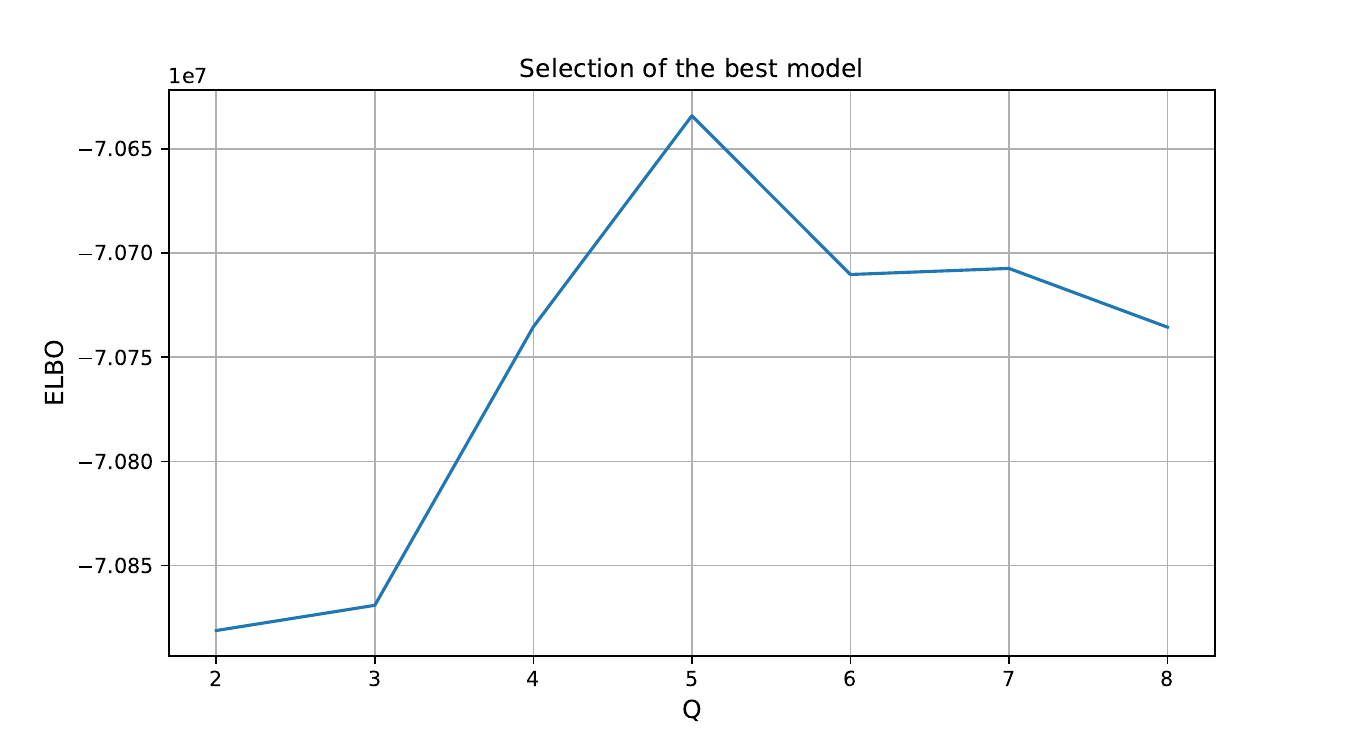}
	\caption{After running ETSBM with different number of clusters $Q$, the ELBO suggest to keep five clusters.}
	\label{fig:real_data_selection_model}
\end{figure}

\subsection{Results}
The meta-graph presented in Figure \ref{fig:meta_graph} is a high-level representation of the network. The ``meta-nodes'' correspond to ETSBM clusters and the edges to the meta-documents as defined in Equation \eqref{eq:expected_meta_docs}. A translation of the top words is provided in Appendix \ref{fig:meta_topics_english}. It is interesting to note the two types of clusters uncovered. In particular, Cluster 5 is composed of central accounts such as French politicians and their communication teams, for instance \textit{Jean-Luc Mélenchon},\textit{ Guillaume Peltier}, \textit{En Marche \#avecvous}, \textit{les Républicains} or \textit{Eléonore Lhéritier}. Some popular French media such as \textit{BFMTV}, \textit{Le Figaro}, \textit{Valeurs actuelles}, \textit{franceinfo} are also in this cluster. On average, the accounts in this cluster have been retweeted or mentioned 299 times against 12 times for the whole network. This cluster does not correspond to a political trend but to accounts with a high level of interactions with the rest of the graph. Despite the small size of this cluster, composed of 25 nodes, ETSBM is able to detect it and to render its central function as a relay of information to other parts of the graph. This is stressed by Topic 1, the main topic discussed within Cluster 5. It regards the election as a democratic process: ``round", ``vote", ``power", ``president", ``first" which we assume stands for ``first round". This core cluster is retweeted differently by the four other clusters which on the contrary hold clear political trends. Cluster 2 and Cluster 3 are interested in \textit{Jean-Luc Mélenchon} (Topic 2) and left parties in general (Topic 4) but they seem to differ in terms of function. Cluster 2 clearly relays information about \textit{Jean-Luc Mélenchon} and is interacting with Cluster 4, interested in \textit{Eric Zemmour}. On the contrary, Cluster 3 seems to only relegate contents without being retweeted. Eventually, Cluster 4, interested in Eric Zemmour (Topic 5), appears to relegate contents from the central accounts as well as sharing many of its own content. This dynamic differs from Cluster 1 interested in \textit{Emmanuel Macron} (Topic 3), which mainly retransmits informations without many self interactions. 
To conclude, the three-way split of the French political landscape is rightfully captured. ETSBM is also able to detect subtleties such as a split within the left-wing, with the orange cluster interested only in \textit{Jean-Luc Mélenchon} and the biggest one exchanging about different left-political front runners, \textit{Jean-Luc Mélenchon}, \textit{Yannick Jadot},  \textit{Fabien Roussel} and \textit{Anne Hidalgo}. ETSBM combines the connection information, for instance all clusters are connected to Cluster 5, and the topics information, for instance Cluster 2 and Cluster 3 should be separated, to provide relevant insights about the information organisation within the social network. This level of detail is promising and highlights how ETSBM gives a better comprehension of the complex dataset at our disposal. 

\newpage
\newgeometry{bottom=0mm} 
\begin{figure}
	\centering
	\begin{subfigure}{\textwidth}
		\centering
		\includegraphics[width=0.75\textwidth]{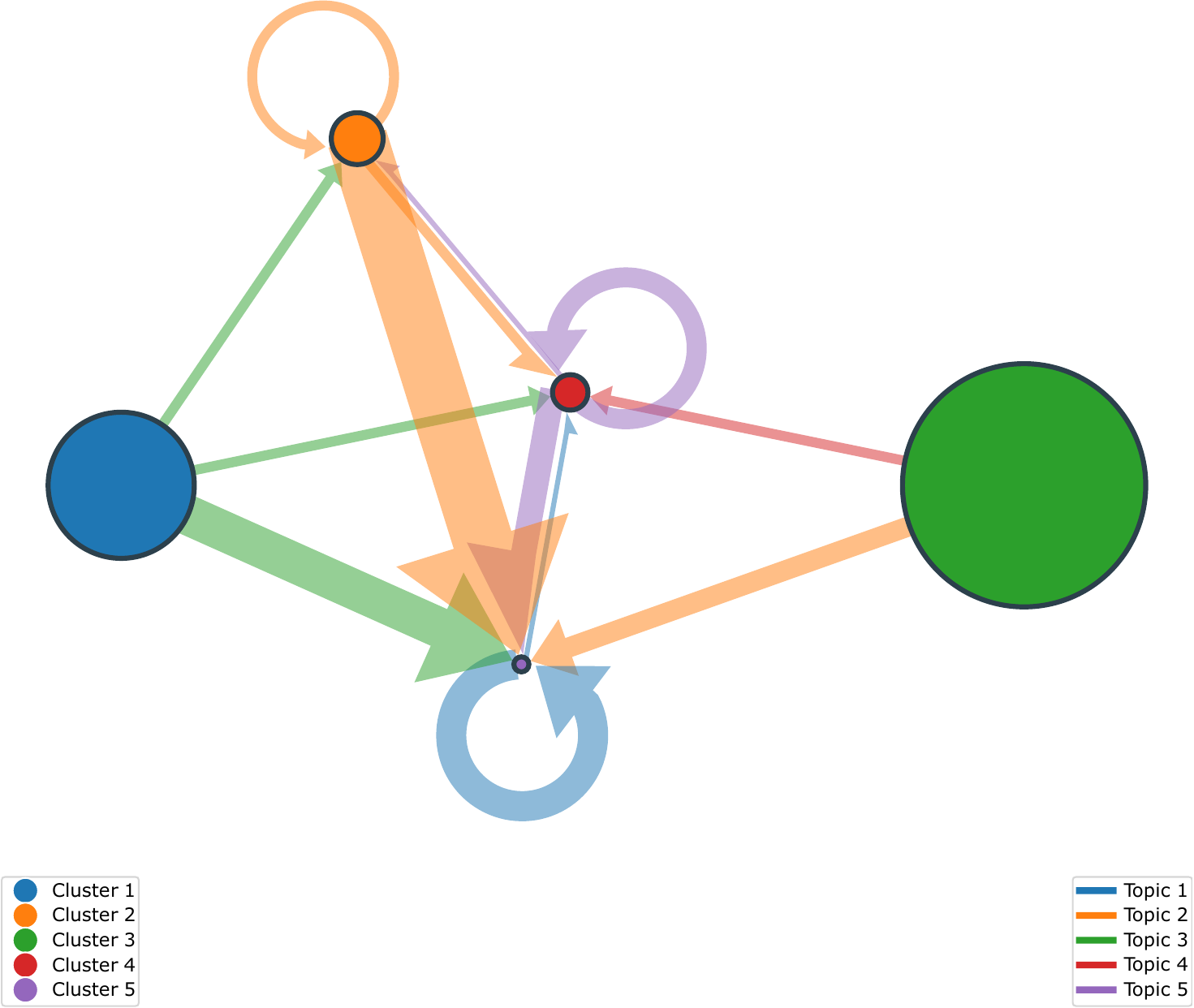}
	\caption{ Meta-network obtained with ETSBM. Each node corresponds to a cluster and the node widths are proportional to the posterior cluster proportions. On the other hand, the edges are coloured as the most used topics within the meta-documents and the widths are proportional to the posterior probabilities of connections between clusters. } 
		\label{fig:meta_graph}
	\end{subfigure}
	\vfill
	\vspace*{0.5cm}
 \begin{subfigure}{\textwidth}
	\centering
	\includegraphics[width=0.9 \textwidth]{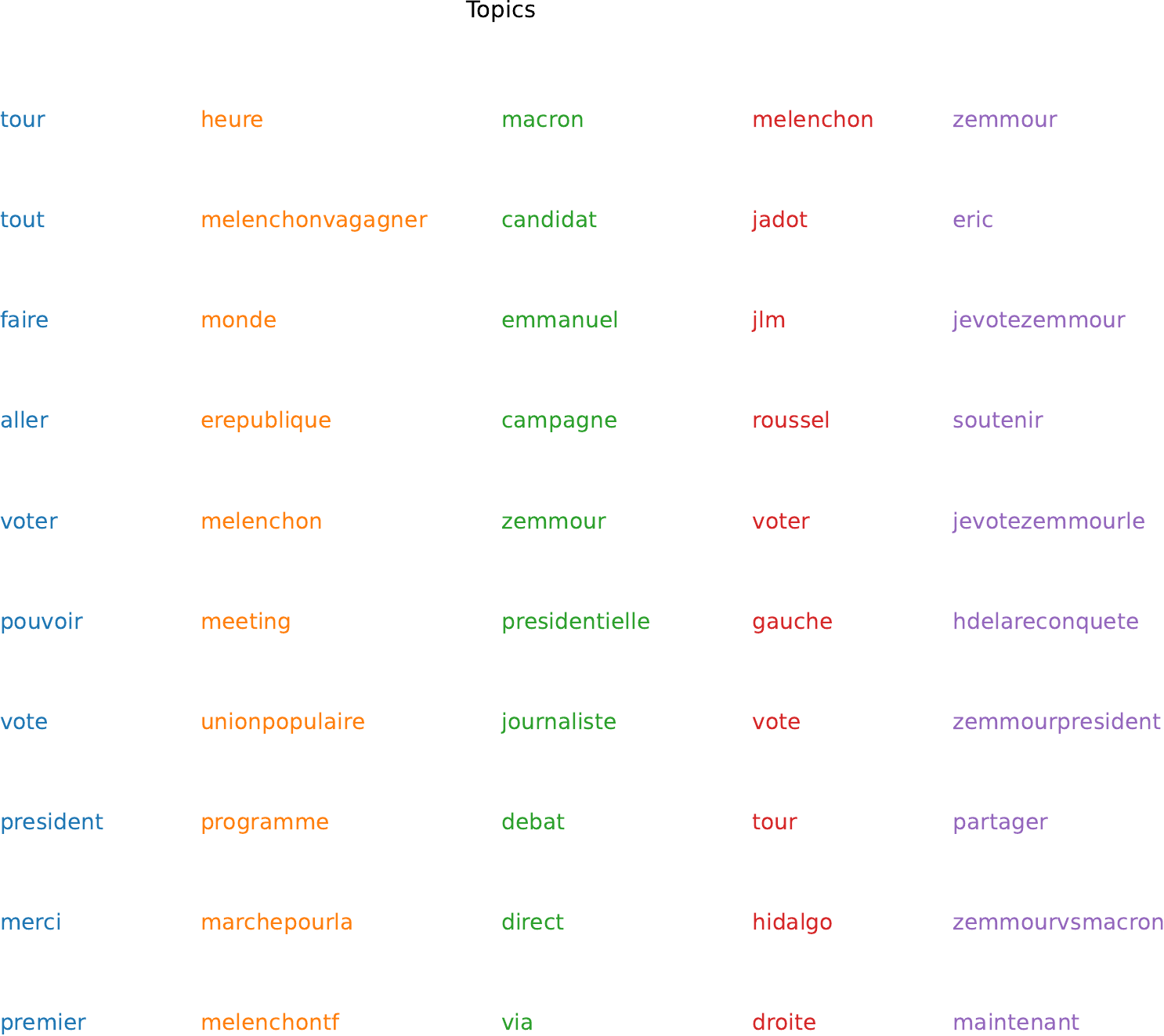}
	\caption{The most important words of the topics presented in the meta-graph above for ETSBM. A translation is provided in Figure \ref{fig:meta_topics_english} of the appendix.}
	\label{fig:meta_topics}
\end{subfigure}
\caption{ETSBM results on the Twitter dataset for $Q=5$ clusters.}
\end{figure}
\restoregeometry

\newpage
\newgeometry{bottom=0mm} 
\subsection{ Comparison with SBM and ETM fitted independently}
\begin{figure}[H]
	\centering
	\begin{subfigure}{\textwidth}
		\centering
	\includegraphics[width=0.7\linewidth]{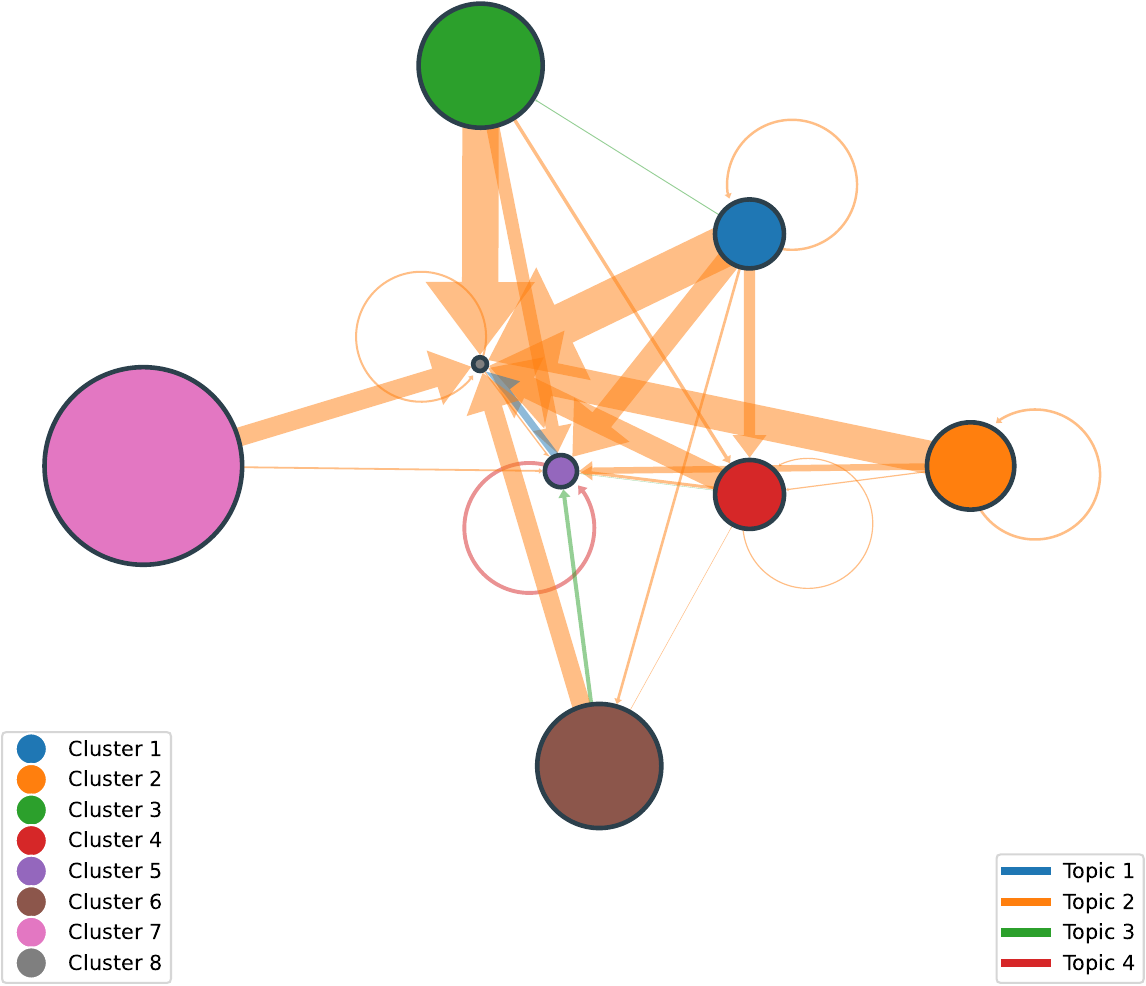}
	\caption{Meta-network estimated with SBM. Each node corresponds to a cluster and the node widths are proportional to the cluster proportions. On the other hand, the edges are coloured as the most used topics of the documents exchanged between the pairs of clusters found by SBM alone. Such topics are obtained by applying ETM alone. The widths of the edges are proportional to the probabilities of connections between clusters.}
	\label{fig:meta_network_sbm_etm}
	\end{subfigure}
	\vfill
	\vspace*{0.5cm}
\begin{subfigure}{\textwidth}
\centering
	\includegraphics[width=0.9\linewidth]{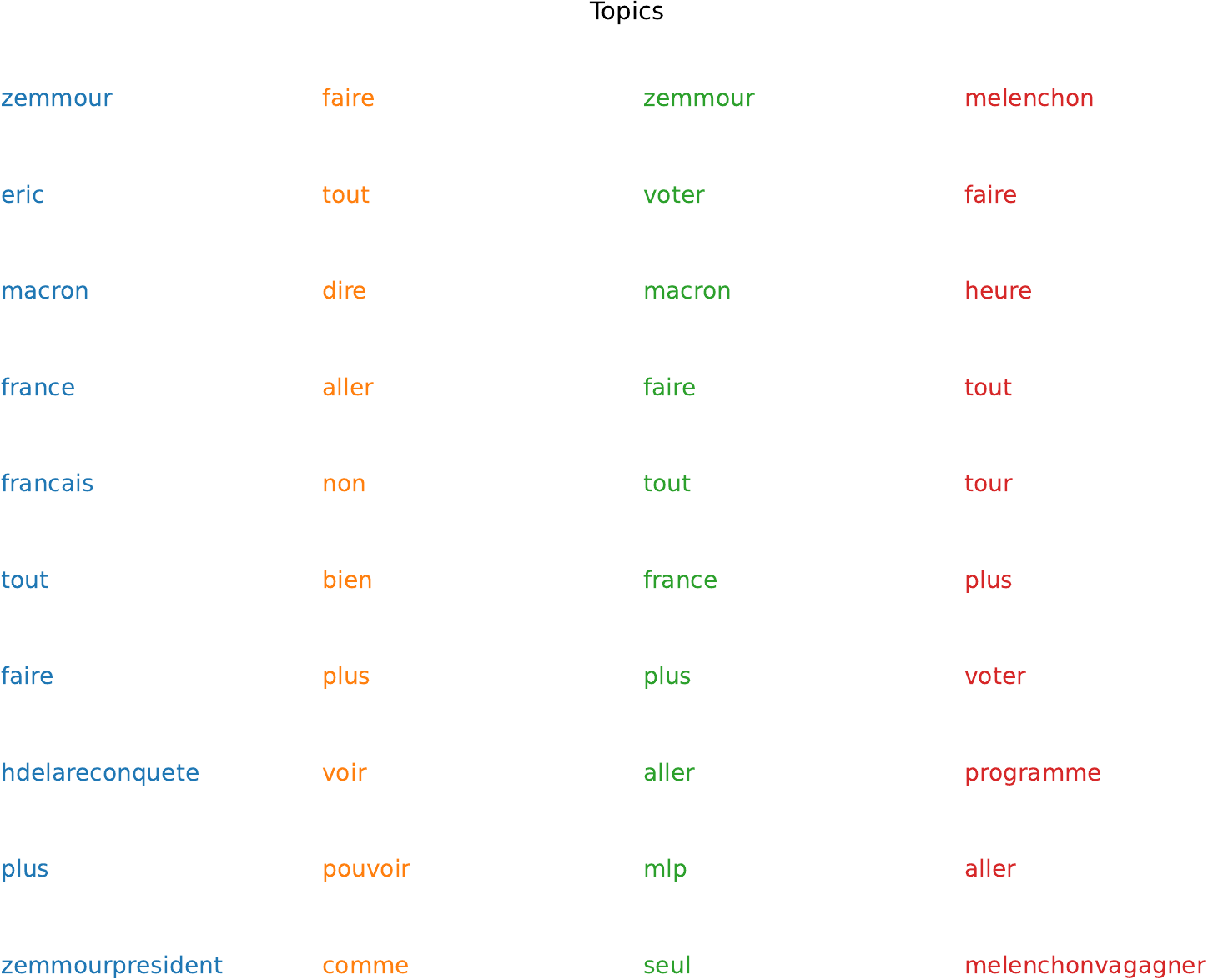}
	\caption{Meta-topics estimated with ETM on the Twitter dataset. A translation is provided in Figure \ref{fig:meta_topics_etm_sbm_english} of the appendix.}
	\label{fig:meta_topics_sbm_etm}
\end{subfigure}
\caption{SBM and ETM results on the Twitter dataset for $Q=8$ clusters.}
\label{fig:meta_results_sbm_etm}
\end{figure}
\restoregeometry

 \paragraph{Description of the results}
 We now give the results obtained using SBM and ETM independently on the Twitter dataset in Figure \ref{fig:meta_results_sbm_etm}. The number of topics is set to $K=20$ again, but only the ones appearing in the meta-graph are presented.  As in the previous section, we restrict the search of the number of clusters between 2 and 8 to keep the results easily interpretable and to provide a fair comparison with ETSBM. The ICL criterion selects a number of clusters $Q=8$ which is the maximum value considered. SBM detects a central cluster in terms of connectivity of the graph (cluster 8), such that all other clusters are connected to it. It is composed of two accounts, the BMFTV account as well as Jean-Luc Mélenchon account. Most connections are dealing with Topic 2, which is very general but not informative.
 
 \paragraph{Comparison with ETSBM results}

The topics in Figure \ref{fig:meta_results_sbm_etm} do not provide much information to understand the content of the connections in the network. In particular, Topic 2, which is general and not specific, is the most used topic in the meta-network. This can be explained by the independence between the construction of the clusters and of the content of the tweets. Therefore, the meta-documents exchanged between clusters have no reason to be specific or to share a common topic. As a result, the Topic 2 emerges as the most used topic between clusters. Compared to ETSBM results, the connections are not informative and the topics exchanged are too general to be considered for interpretation. We emphasise that among the 20 topics estimated by ETM, some are very informative but do not emerge in the meta-graph, backing the claim that the clusters are not meaningful. In addition, the number of clusters selected by the ICL ($8$), is higher than the number of clusters selected by ETSBM ($5$). Having a low number of clusters can help make the results easier to understand.

\section{Conclusion and discussion} \label{sec:conc_disc}

The embedded topics for the stochastic block model (ETSBM) is well suited to simultaneously find meaningful node and edge clusters. In addition, ETSBM provides an intelligible high-level representation of the graph. It can be used both on directed and undirected graphs and is suited for large datasets thanks to the variational inference. The numerical experiments showed that the ELBO is a relevant model selection criterion to estimate the number of node clusters $Q$ in this Bayesian framework. Moreover, this  criterion keeps provide a good estimate of $Q$ for a high number of topics $K$. In the end, a use case on a Twitter dataset proved the usefulness of the method. ETSBM clustering results were both meaningful and humanly intelligible. Further work may be directed in the study of theoretical foundations of the model selection criterion proposed. Adding temporal information concerning the connectivity patterns and the topics modelling could also contribute to obtain useful information on the data.

\bibliography{biblio}

\newpage
\appendix

\section{Inference}
\begin{proof}[Proof of Proposition \ref{prop:elbo_decomposition}]\label{appendix:elbo_decomp}
	
The ELBO can be decomposed as follow:
\begin{align*}
	\log  p(A, W \mid \alpha, \rho) & = \mathbb{E}_{R} \left[ \log  p(A, W \mid \alpha, \rho)\right] \\
	& = \mathbb{E}_{R} \left[  \log \frac{p(A, W, Y, \pi , \gamma, \delta  \mid \alpha, \rho)}{p(Y, \pi , \gamma, \delta \mid A, W, \alpha, \rho)}   \right] \ \tag*{applying Bayes rule}\\
	& = \mathbb{E}_{R} \left[  \log \frac{p(A, W, Y, \pi , \gamma, \delta  \mid \alpha, \rho)}{R(Y, \pi , \gamma, \delta)} + \log \frac{R(Y, \pi , \gamma, \delta)}{p(Y, \pi , \gamma, \delta \mid A, W, \alpha, \rho)}   \right] \\
	& =  \mathscr{L} (R(\cdot); \alpha, \rho ) + \KL(R(\cdot) || p(Y, \pi, \gamma, \delta \mid A, W, \alpha, \rho)).
\end{align*}
\end{proof}

\begin{proof}[Proof of Proposition \ref{prop:elbo_parametric_function}]\label{prop:elbo_detail_proof}
\begin{align} \label{eq:elbo_detail}
	\mathscr{L}(R(\cdot); \alpha, \rho ) 	& = \overset{ \mathscr{L}^{net}(\tau, \tilde{\pi}_{qr1}, \tilde{\pi}_{qr2}  \tilde{\gamma} ; \alpha, \rho ) :=}{ \overbrace{\mathbb{E}_{R}\left[ \log \frac{p( W \mid Y, A, \theta, \alpha, \rho ) p(\theta)}{ R(\theta)}\right]}} +  \overset{ \mathscr{L}^{texts}(\tau, \nu ; \alpha, \rho ) :=}{ \overbrace{ \mathbb{E}_{R}\left[ \log \frac{ p(A \mid Y, \pi) p(Y \mid \gamma)  p(\pi) p(\gamma)}{ R(Y) R(\pi) R(\gamma)} \right]}} \nonumber\\
	& = \mathbb{E}_{R}\left[ \log p( W \mid Y, A, \theta, \alpha, \rho ) \right]+  \mathbb{E}_{R}\left[ \log p(\theta) \right] -  \mathbb{E}_{R}\left[ \log R(\theta) \right] \nonumber \\
	& + \mathbb{E}_{R}\left[ \log p(A \mid Y, \pi)  \right] + \mathbb{E}_{R}\left[\log p(Y \mid \gamma) \right] + \mathbb{E}_{R}\left[ \log  p(\pi) \right] + \mathbb{E}_{R}\left[ \log p(\gamma) \right] \nonumber\\
	& - \mathbb{E}_{R}\left[\log R(Y) \right] - \mathbb{E}_{R}\left[ \log R(\pi) \right] - \mathbb{E}_{R}\left[ \log R(\gamma) \right] \nonumber\\
	& =  \sum_{i \neq j }^M \sum_{ q, r}^Q A_{ij}  \tau_{iq} \tau_{jr}  \mathbb{E}_{R}\left[  \underset{T_{ij}^{\delta_{qr}}}{\underbrace{ \log p(w_{ij} \mid \delta_{qr}, \alpha, \rho )}} \right] - \sum_{q,r} \KL( \mathcal{N}(\mu_{qr}(\tau, \nu), \sigma_{qr}(\tau, \nu)) || \mathcal{N}(0, I) ) \nonumber\\
	& + \sum_{i \neq j}^M \sum_{q,r}^Q  \tau_{iq} \tau_{jr}  A_{ij} \left( \psi( \kappa_{qr1}) - \psi( \kappa_{qr2}) \right)  +  \sum_{i \neq j}^M \sum_{q,r}^Q\tau_{iq} \tau_{jr}  ( \psi(\kappa_{qr2} ) - \psi(\kappa_{qr1} + \kappa_{qr2} )) \nonumber \\
	& + \sum_{i=1}^M \sum_{q=1}^Q \tau_{iq} \left( \psi(\gamma_{q}) -  \psi\left(\sum_{q} \gamma_{q}\right) \right)  + \log \mathcal{B}(\mathrm{1}_Q) + \log(\mathcal{B}(a,b) ) \nonumber \\
	&  - \sum_{i=1}^M \sum_{q=1}^Q  \tau_{iq} \log(\tau_{iq})  -  \sum_{q,r} \log \mathcal{B}(\kappa_{qr1}, \kappa_{qr2}) - \log \mathcal{B}(\gamma).
\end{align}
where,
\begin{align}
	T_{ij}^{\delta_{qr}} = \sum_{d=1}^{D_{ij}} \sum_{n=1}^{N_{id}^d} \sum_{v=1}^V w_{ij}^{dnv} \log\left( \sum_{k=1}^K  \theta_{qr k}   \beta_{k v} \right).
\end{align} 
and $\theta_{qr} = \mu_{qr}(\tau, \nu) + \sigma_{qr}(\tau, \nu) \epsilon$, 
$\epsilon \sim \mathcal{N}(0_K, \mathrm{I}_{K})$.

The Kullback-Leibler divergence between two Gaussian variables has a close form and is easy to compute. All the terms can be computed except for the expectation of $	T_{ij}^{\delta_{qr}}$ that can be approximated using a Monte-Carlo estimator, by drawing $S$ samples for each pair $(q,r)$, such that:
\begin{align*}
	\ \ \epsilon^s \sim \mathcal{N}(0,I_{K}),\ \ \ \delta_{qr}^s = \mu_{qr}( \tau, \nu) + \sigma_{qr}( \tau, \nu) \odot \epsilon^s,\ \ \theta_{qr}^s = \softmax(\delta_{qr}^s).
\end{align*}
with $ \odot$ denoting the Hadamard product. Thus, for each pair of nodes $(i,j)$ and pair of clusters $(q,r)$, the estimate is given by: 
\begin{align*}
	\hat{T}_{ij}^{qr} = S^{-1} \sum_{s=1}^S T_{ij}^{\delta^s_{qr}}.
\end{align*}
Plugging $\hat{T}_{ij}^{qr}$ in the Equation \eqref{eq:elbo_detail} gives the final estimator of the ELBO.

\end{proof}
\newpage
Figure \ref{fig:meta_topics_english} provides a translation of topics found by ETSBM on the real dataset and appearing in the meta-network.
\section{Real data}

 \begin{figure}[ht]
	\centering
	\includegraphics[width=\textwidth]{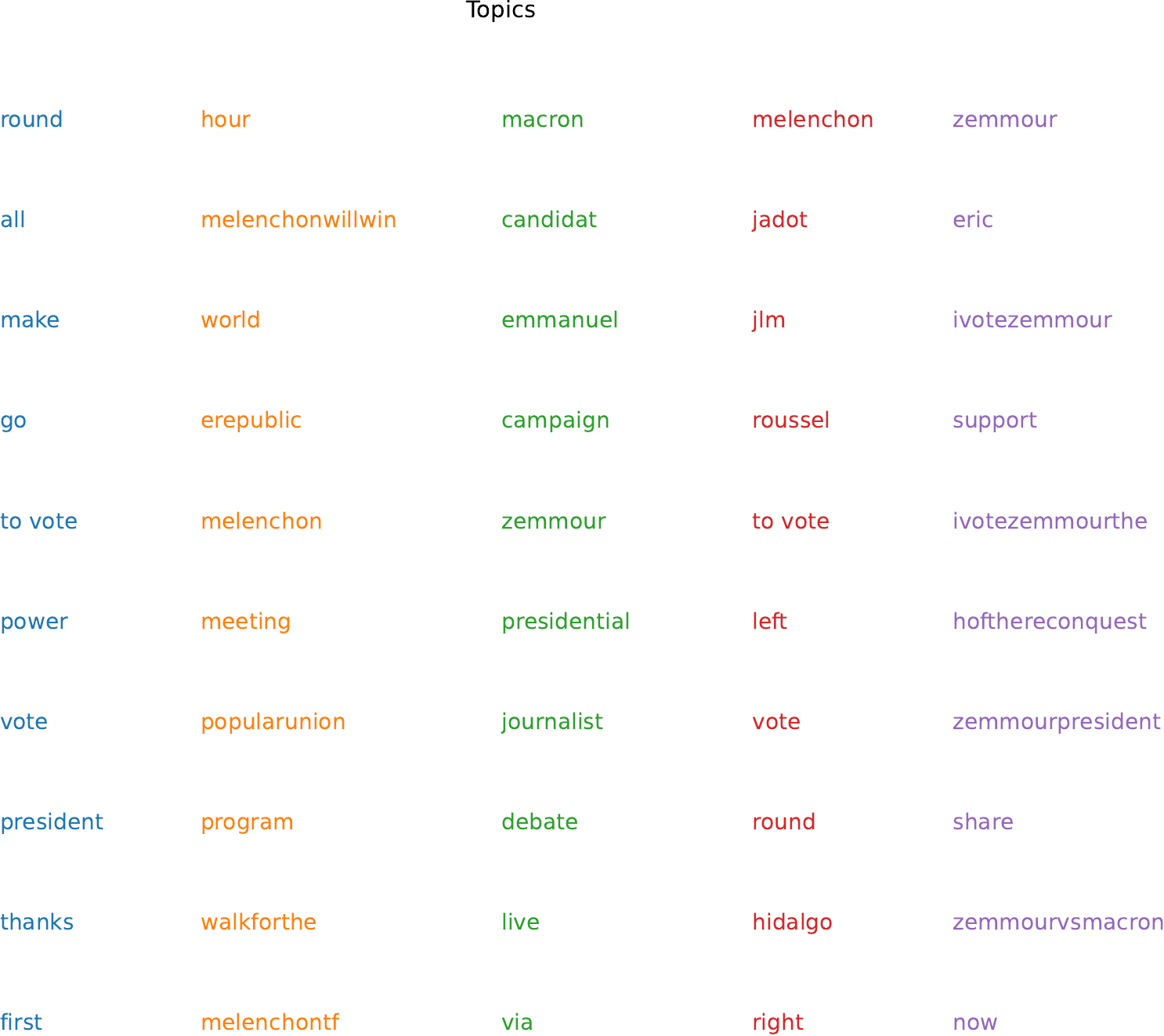}
	\caption{The most important words of each topic present in the meta-graph translated in English.}
	\label{fig:meta_topics_english}
\end{figure}

\newpage
Figure \ref{fig:meta_topics_etm_sbm_english} provides a translation of topics found by ETM on the real dataset and appearing in the meta-network.
\begin{figure}[ht]
	\centering
	\includegraphics[width=\textwidth]{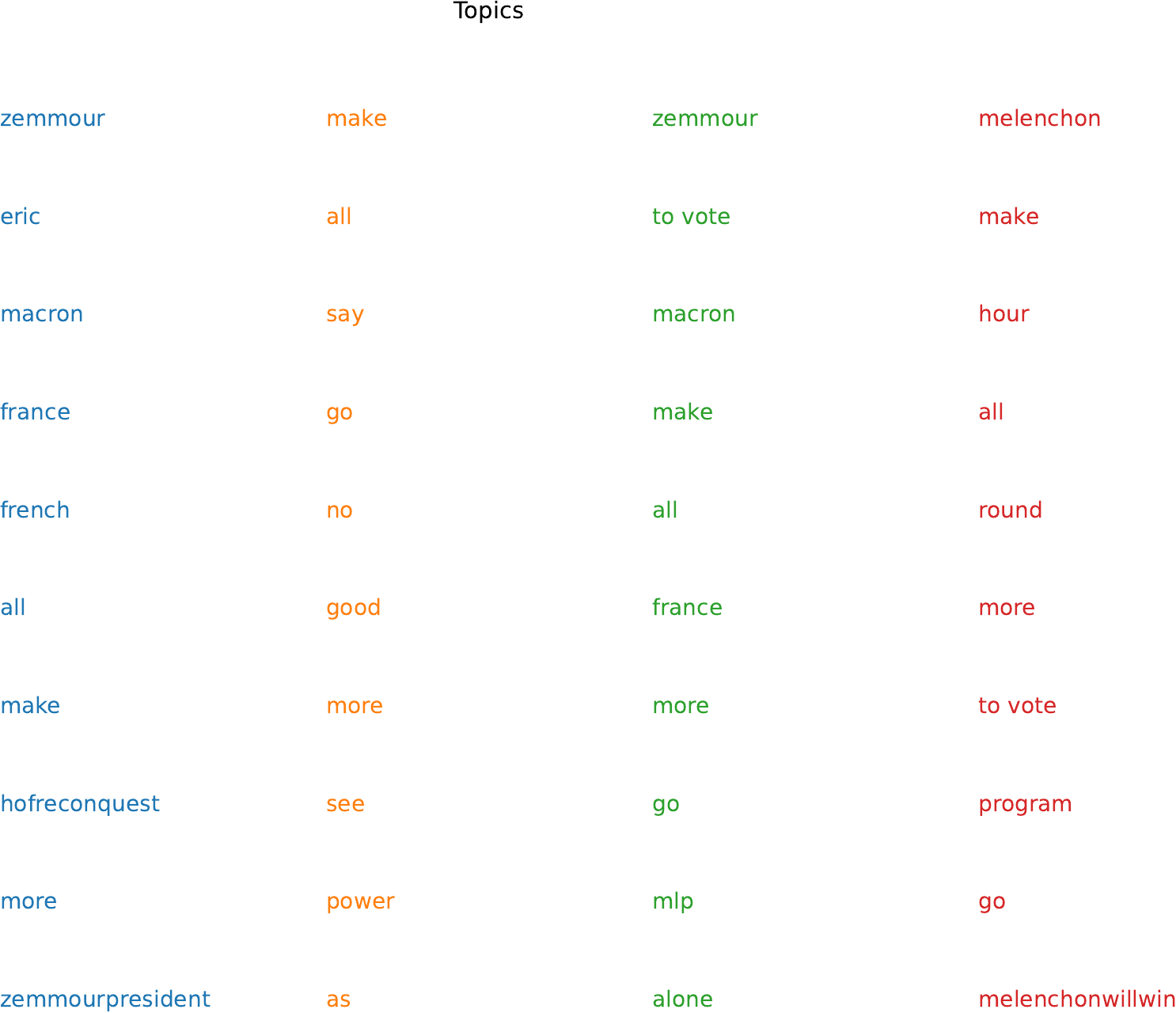}
	\caption{The most important words of each topic present in the meta-graph translated in English.}
	\label{fig:meta_topics_etm_sbm_english}
\end{figure}

\end{document}